%% file: main.tex
\documentclass{article}

\usepackage{PRIMEarxiv}

\usepackage[utf8]{inputenc} 
\usepackage[T1]{fontenc}    
\usepackage{url}            
\usepackage{dcolumn}
\usepackage{booktabs}       
\usepackage{multirow}
\usepackage{amsfonts}       
\usepackage{nicefrac}       
\usepackage{microtype}      
\usepackage{lipsum}
\usepackage{fancyhdr}       
\usepackage{xcolor}         
\usepackage{graphicx}       
\usepackage{amsmath}
\usepackage{amsthm}
\usepackage{overpic}
\usepackage{helvet}
\usepackage{siunitx}
\usepackage{array}
\usepackage[numbers,sort&compress,round]{natbib}
\usepackage{mathrsfs}
\usepackage{soul}
\usepackage{ulem}
\usepackage{svg}
\usepackage{comment}
\usepackage{colortbl}
\usepackage{pdflscape}
\usepackage{xspace} 
\usepackage{sectsty} 
\usepackage[toc,page,header]{appendix}
\usepackage{minitoc}
\usepackage{caption}
\DeclareCaptionLabelSeparator{pipe}{ | }
\definecolor{darkblue}{HTML}{627EBC}
\definecolor{darkred}{HTML}{C0483A}
\usepackage[breaklinks,colorlinks=true,citecolor=darkblue,linkcolor=darkred,urlcolor=darkblue]{hyperref}

\allsectionsfont{\sffamily\bfseries}
\setcitestyle{numbers,round,comma}

\theoremstyle{plain}
\newtheorem{proposition}{Proposition}

\theoremstyle{remark}

\theoremstyle{definition}
\newtheorem{definition}{Definition}
\theoremstyle{plain}
\newtheorem{lemma}{Lemma}

\newcommand{\chris}[1]{\textcolor{teal}{[Chris: #1]}}

\excludecomment{mycomment}

\newcommand{\FW}{generative network game graph partner-choice framework\xspace}
\newcommand{\oursfull}{Game Theory Block Model\xspace}
\newcommand{\ours}{GTBM\xspace}

\newcommand{\bA}{\ensuremath{\mathbf{A}}\xspace} 
\newcommand{\bM}{\ensuremath{\mathbf{M}}\xspace} 
\newcommand{\bB}{\ensuremath{\mathbf{B}}\xspace} 
\newcommand{\bU}{\ensuremath{\mathbf{U}}\xspace} 

\newcolumntype{R}{>{\raggedleft\arraybackslash}p{1.6cm}}
\newcolumntype{C}[1]{>{\centering\arraybackslash}m{#1}}

\title{Strategic Origins of Network Structure}

\author{Christoph Riedl \\
  D'Amore-McKim School of Business,\\
  Khoury College of Computer Sciences, \& \\
  Network Science Institute\\
  Northeastern University\\
  Boston, MA 02115 \\
  \texttt{c.riedl@northeastern.edu} \\
\And
  Rory Smead \\
  College of Social Sciences and Humanities \\
Northeastern University\\
Boston, MA 02115\\
    \texttt{r.smead@northeastern.edu}
}

\begin{document}
\doparttoc 
\faketableofcontents 

\part{} 

\maketitle

\begin{abstract}
The structure and dynamics of complex networks are intimately linked.
Despite immense progress in statistically characterizing these systems, a generalized mechanistic framework explaining why specific mesoscale structures emerge remains elusive. 
Here we introduce the game-graph equilibrium, the stable rest point of a payoff-driven partner-choice dynamic, analogous to a Nash equilibrium. This reveals network structure as the fossilized record of strategic interaction.
We prove analytically that two-type games have exactly four network topologies, which are their exhaustive thermodynamic phases: assortative, disassortative, and two core-peripheries. This result generalizes to $K^K$ topological classes for $K$-type interactions. Furthermore, four distinct Nash equilibria produce core-periphery topologies, yet their payoffs reveal whether the underlying dynamic is cooperative or exploitative. To recover the generating dynamics from observed networks, we develop mechanistic inference by constraining a block model to lie on a low-dimensional, game-theoretic manifold. 
Applied to real-world networks, the model yields interpretable quantities, distinguishes states that look identical to conventional methods, predicts network topologies it has never seen without retraining, and controls its own complexity without hand-tuned penalties.
Networks such as international trade are well described by the model despite the hard constraint, suggesting that real-world systems are organized by payoff-driven partner choice.
Our framework resolves a foundational ambiguity by offering a precise, mechanistic definition of communities: strategic equivalence.
\end{abstract}


\begin{mycomment}
\clearpage
\section*{CONTRIBUTIONS}
\begin{itemize}
    \item Ideas for reviewers / editors: Duncan Watts, Steve Strogatz, M. E. J. Newman, Matt Jackson
    \item unified generative framework that arbitrates between competing sociological and network science paradigms; partner-choice functional form (frequency-dependent softmax)
    \item complete reframing of prior theory: agents don't choose strategies; agents choose partners; changing types is hard
    \item game theory directly maps onto network types: (a-b) vs. (c-d) translates into (a-c) vs. (b-d) in network space
    \item discover and analytically characterize phase transitions between game-theoretic payoffs and macroscopic network topologies (Phase Diagram of Network Topologies); characterize parameter recovery, susceptibility
    \item all possible games produce only 4 different types of graphs: pair X, pair Y, pair same, pair different; i.e., In standard network science, topologies (assortative, disassortative, core-periphery, bipartite) are treated as descriptive labels.
    \item we can predict (forecasting): Using a longitudinal dataset of global networks, we demonstrate that our model successfully predicts future network topologies given only the new macroscopic node proportions.
    \item the fact that forecasting works is itself proof for the theory that complex systems (like trade) operate according to our theorized dynamics
    \item Semantic interpretability and mechanistic explanation: SBM only find descriptive clusters but cannot explain why they arise or what they mean; we can assign game labels and payoffs
    \item {Bayesian SBMs provide statistically principled descriptions of mesoscale structure but leave the block matrix mechanistically uninterpreted.}
    \item \chris{Channel tiago: You can fairly say that standard SBMs don't provide mechanistic or causal interpretation of the block matrix — but calling them "only descriptive" in a paper that might land on his desk is picking a fight with his exact terminology. Better framing: "Standard SBMs infer statistically warranted structure but leave the block matrix uninterpreted — it tells you that groups connect with certain densities, but not why.}
    \item universal and extendable to 3x3, and NxN
    \item we solve label invariance: standard SBM are unsupervised and detected cluster labels are permutation invariant; we develop symmetry breaking that permanently pins the labels to the high-affinity group.
    \item implement as constrained SBM; both 2x2 and 3x3
    \item implicit model selection: 3x3 SBM returns 2x2 solutions if there is one
    \item implemented practical tests for hypothesis test: is the system explained by game-theoretic dynamics or not
    \item $\ell$ acts to quantify how rational the system is
    \item demonstrate practical application to datasets across domains: social, economic, biological
    \item interesting insights about trade networks
    \item very cool connection with geometric networks: geometric networks have manifold of spatial distance, we have manifold of payoff distance; future work could combine them
    \item important implications for policy: if you want to change the network, make sure you change the payoffs far enough to actually cross phase transition boundaries
\end{itemize}
\end{mycomment}


\clearpage

Complex systems display a remarkable universality in their architecture, consistently organizing into topologies defined by distinct structural properties \cite{erdos1959random,watts1998collective,barabasi1999emergence}. This topology interacts with the dynamic behaviors that unfold on it, governing collective behaviors such as synchronization \cite{strogatz2000kuramoto}, spreading processes \cite{pastor2001epidemic,lee2025complex}, systemic resilience \cite{gao2016universal}, and formation of social consensus \cite{fraiberger2018quantifying}. Indeed, the coupling between structure and dynamics is so tight that interaction patterns can reveal the underlying architecture \cite{ma2025reconstructing}, just as architecture can predict interactions \cite{kovacs2019network}. Yet, despite immense progress in statistically characterizing these systems \citep{barzel2013universality,fortunato2010community,jackson1996strategic}, a generalized mechanistic framework explaining why specific mesoscale structures emerge remains elusive. 

Established approaches, such as stochastic block models and community detection methods, provide statistically principled descriptions of mesoscale structure, but leave the block matrix mechanistically uninterpreted. They tell us that groups connect with certain densities, not why \citep{karrer2011stochastic, peixoto2017nonparametric}. 
Micro-foundational approaches explain the mechanisms of network evolution but lack generalizability \citep{barabasi1999emergence,watts1998collective}, as they are typically fine-tuned to produce single structural patterns (e.g., small-world or scale-free networks).
Game theory offers a natural language for interaction incentives, but existing approaches either derive behavioral outcomes from a given topology \citep{ohtsuki2006simple, szabo2007evolutionary}, or treat network formation as a strategic game of link choice \citep{jackson1996strategic, bala2000noncooperative}. Coevolutionary models let games and topology evolve jointly \citep{skyrms2000dynamic, santos2006cooperation}, yet often remain simulation-based and game-specific \citep{fulker2021spite}. What remains absent is the inverse mapping: a general analytical link from the parameters of an arbitrary game to the equilibrium mesoscale structure it induces. Consequently, the mechanistic side cannot scale to the complex, noisy topologies that statistical models so easily describe (see SI \ref{si:related_work} for more related work).

Here we formalize the connection between game payoffs and network topology by inverting game theory's central question. Classical game theory asks whether anyone could gain by changing strategies; we ask whether anyone could gain by changing partners. A game-graph equilibrium, analogous to a Nash equilibrium, is a state where no partner-changing will improve payoffs assuming fixed strategies. These states will be stable rest points of partner-choice dynamics and be a fossilized record of the strategic interaction (Figure~\ref{fig:panel1_conceptual}). 
We prove analytically that the four network topologies---assortative, disassortative, and two core-peripheries---are in fact the exhaustive thermodynamic phases of an evolutionary game. We generalize this result to $K^K$ topological classes for $K$-type interactions. Furthermore, four distinct Nash equilibria all produce core-periphery networks, yet their payoffs reveal whether exploitation or cooperation is at work, a distinction that is not available from standard block models. 
To recover the generating dynamics from observed networks, we develop mechanistic inference by constraining a block model to lie on a low-dimensional, game-theoretic manifold for two- and three-type interactions. 

We apply our inference framework to empirical networks across social, economic, and biological networks. These applications demonstrate how the model yields representations that are directly interpretable in terms of costs and benefits, and roles like ``producers'' vs.~``exploiters.'' It also resolves degeneracies that trip up traditional models, forecasts entirely new regimes out of the box, and avoids overfitting without external regularization.
We find networks like international trade are well described by the model, despite the hard constraint, suggesting that real-world systems are organized by payoff-driven partner choice.
In doing so our model resolves a foundational ambiguity about what a network community is: nodes that occupy strategically equivalent roles. It also pins down the number of communities a network has as equal to the distinct roles, a number the model arrives at endogenously rather than one the analyst must choose or test.
Crucially, misfit is itself a finding. The model encodes exactly the structure that pairwise strategic choice can produce, so when it fails to fit a network, it indicates organization that exceeds what pairwise interaction can generate. Thus, misfit is a falsifiable signal of higher-order, multi-body structure rather than a shortcoming of the method.
Together, these results constitute a unified theory connecting game theory and network science: a direct, analytical map from the payoffs of an interaction to the mesoscale topology it produces. Because the constraint barely costs fit yet recovers mechanism and forecasts rewiring, payoff-driven partner choice is not merely consistent with real networks but appears to organize them. We finish with an online deplatforming case study that illustrates the difference between altering a network's realized topology and altering the incentives that generate it.
A Python implementation will be freely available at \href{https://github.com/riedlc/GTBM}{GitHub}.

\begin{figure}[h!]
    \centering
    \includegraphics[width=1\linewidth]{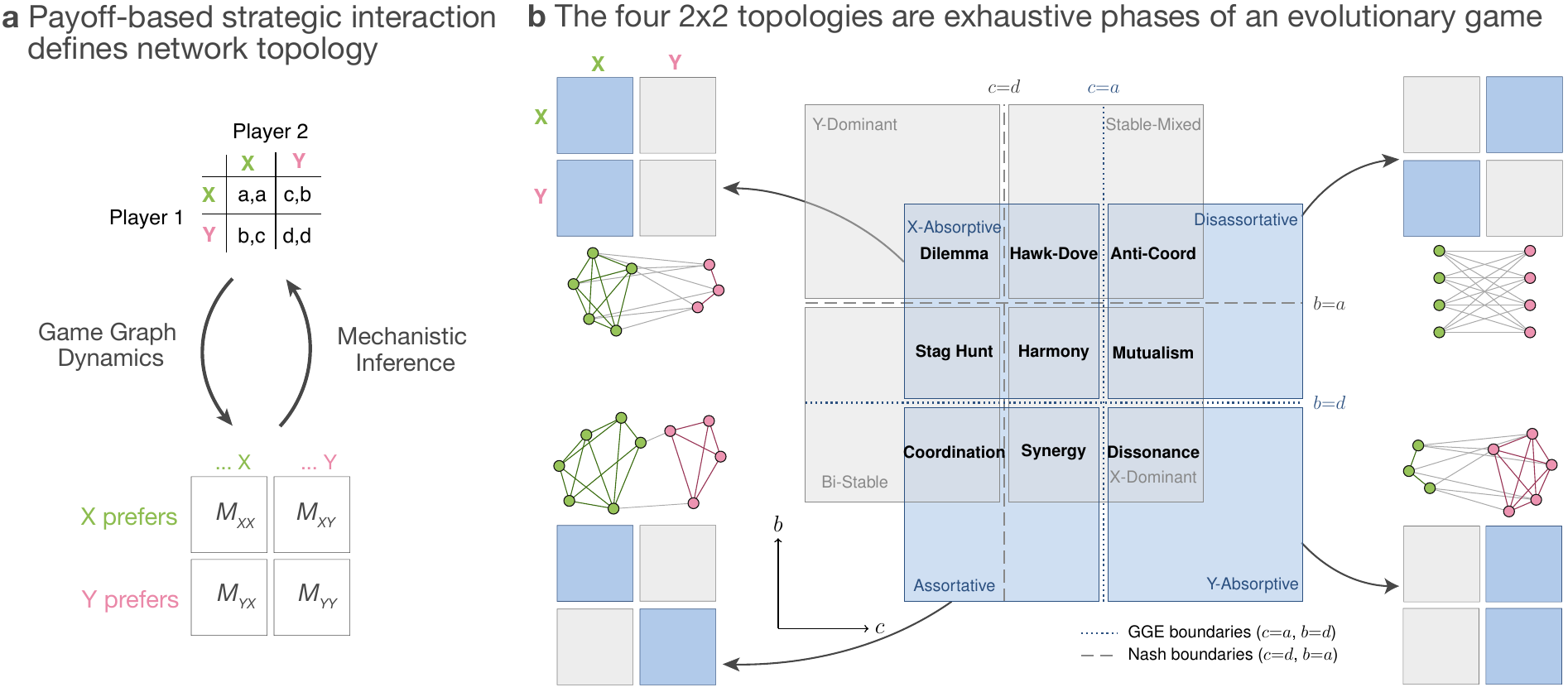}
    \caption{
        \textbf{Connecting payoff-based interaction with evolutionary network topology.}
        \textbf{a} Game-theory payoffs convert into block matrix of connection probabilities via Game-Graph Dynamics. Mechanistic inference identifies payoffs from empirical graphs via the \ours.
        \textbf{b} The network space of Game-Graph Equilibria intersects with the game-theory space of four Nash equilibria (Table \ref{SI:roles} gives semantic role labels for each game).
        }
    \label{fig:panel1_conceptual}
\end{figure}

\section*{Theoretical Model}
\label{model}

We model networks that result from strategic interactions as a process of \textit{partner choice} rather than \textit{strategy choice}. A \textit{game-graph equilibrium} (GGE) is a network configuration such that no individual can increase their payoff by redirecting an outgoing link to a different partner type. GGEs identify which network topologies are stable for which games under any process that favors higher-payoff partners, independently of population frequencies or specific update rules. The dynamics of the model introduces a direct, deterministic mapping between game-theoretic payoffs and resulting network topology, specified by a block matrix \bM in which $M_{ij}$ is the joint probability that a type-$i$ individual interacts with a type-$j$ individual.

A population consists of strategy types defined by their payoff relationships in a game (Figure~\ref{fig:panel1_conceptual}a). Types are assumed to be fixed, but interaction probabilities between types are dynamic. Individuals control their outgoing links and preferentially direct them toward partners that yield higher payoffs. Network ties are updated iteratively: individuals begin from random mixing and update their outgoing link probabilities in proportion to the payoffs received from each partner type (see Methods). This is the replicator dynamic \citep{hofbauer1998evolutionary} applied to link formation rather than strategy frequencies and is analogous to individual-level reinforcement learning \citep{skyrms2000dynamic}. The parameter $\ell$ controls how many rounds of updates occur, with higher $\ell$ producing sharper discrimination among types. The resulting \bM is a coarse-grained representation of the individual-level adjacency matrix \bA in the precise sense of network renormalization \citep{gabrielli2025network}, where a node partition induces a surjective map between adjacency matrices at successive resolution levels. Here, the role played by geometric or diffusional proximity in other schemes \citep{krioukov2010hyperbolic} is played by strategic equivalence. Since all individuals of the same type face identical incentives, any individual-level fixed point is a type-level fixed point: if no individual wants to re-direct a link, the block matrix \bM is at rest.

GGEs are entirely determined by the ordinal payoff structure of the game. For any $K\times K$ game with a unique best partner for each type, there are $K^K$ distinct equilibrium topologies (Proposition~\ref{si:labeled_count}). The GGE's are also the asymptotically stable rest points of any monotone selection dynamic (Proposition~\ref{si:prop-kgt2-hzero}). In $2\times 2$ games, the interaction payoffs ($a, b, c, d$) determine the mesoscale connectivity $(M_{XX}, M_{XY}, M_{YX}, M_{YY})$ (Methods, Eq.~\ref{eq:physics}). The payoff space is partitioned by two phase boundaries ($a=c$ and $b=d$) into four regions, each corresponding to a distinct macroscopic equilibrium topology: assortative, disassortative, and two asymmetric core-periphery configurations (Proposition~\ref{si:result2}). This four-topology classification holds for any monotone selection dynamic (Propositions \ref{si:result1} and \ref{si:result2}). A related but distinct pair of boundaries ($a=b$ and $c=d$) partitions the same space by classical Nash equilibrium structure (Figure~\ref{fig:panel1_conceptual}b) \cite{weibull1997evolutionary}. The combination of these boundaries creates a taxonomy of interaction types that is richer, with nine distinct $2 \times 2$ game types, than we can get from either game-theory or network science (SI Table~\ref{SI:roles} and Figure \ref{fig:affinity_imbalance}). 

\begin{figure}[h!]
    \centering
    \includegraphics[width=1\linewidth]{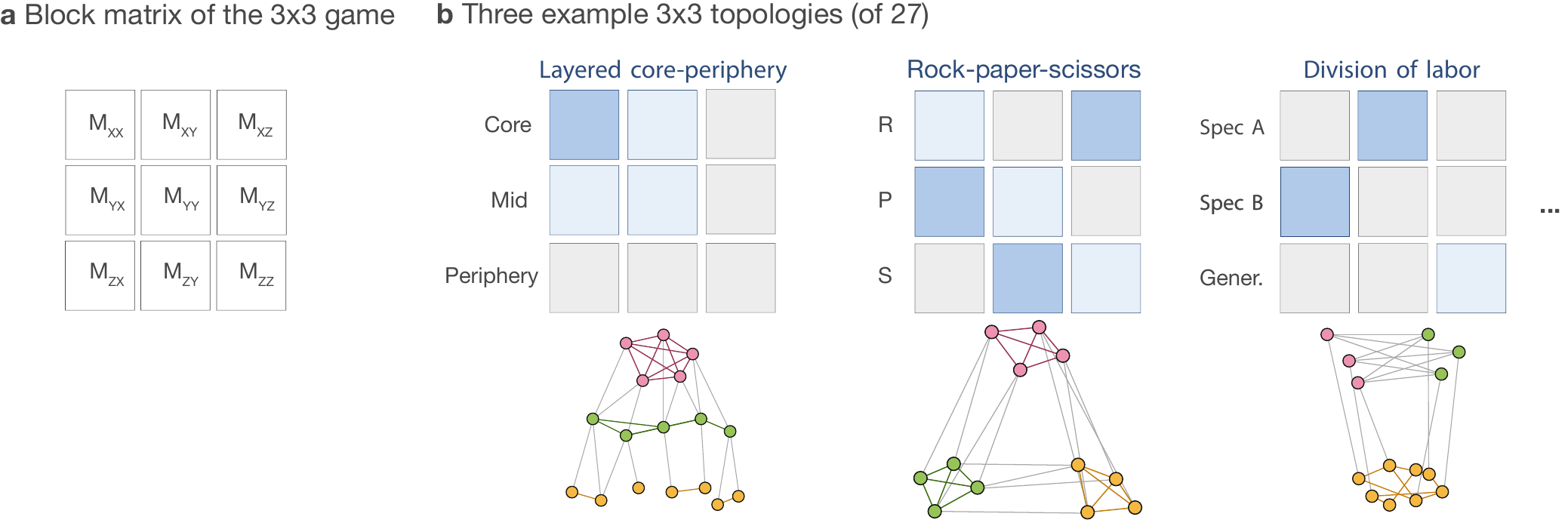}
    \caption{
        \textbf{Extension to the 3x3 game.}
        \textbf{a} Block matrix of 3-type networks.
        \textbf{b} Three example 3x3 network topologies.
        }
    \label{fig:3x3_extension}
\end{figure}

\paragraph{The Rationality Parameter as Inverse Temperature.} Our framework models partner choice via a frequency-dependent softmax activation, mathematically equivalent to the Boltzmann distribution in statistical mechanics and the quantal response function in discrete choice theory \citep{mcfadden1974measurement}. Both formulations describe utility maximization under bounded rationality, governed by an inverse temperature parameter, $\ell$, that controls the signal-to-noise ratio of selection. In socio-economic networks, $\ell$ quantifies agent rationality, or how sharply actors discriminate between payoffs. In biological and chemical systems, where payoffs represent evolutionary fitness or binding affinities, low $\ell$ instead captures stochastic drift or ambient thermal noise \citep{sella2005application}. Despite these distinct substantive interpretations, $\ell$ serves a universal structural role: it sets the sharpness of the topological phase boundaries. In the high-temperature limit ($\ell \rightarrow 0$), discrimination vanishes and the system collapses to maximum entropy (uniform \bB; i.e., an Erd\H{o}s--R\'enyi (ER) random graph) regardless of payoffs. In the low-temperature limit ($\ell \rightarrow \infty$), the softmax converges to a step function and the partition into the four equilibrium topologies becomes exact (Figure~\ref{fig:panel2_physics}a). Network science has long treated these as descriptive categories; they are in fact the exhaustive thermodynamic phases of an evolutionary game. The partition function of this Boltzmann system (Eq.~\ref{eq:mixing_general}) is $Z_r=\sum_s x_s U_{rs}^{\ell}$, and its free energy $F=-\ell^{-1}\sum_r x_r \ln Z_r$ is real-analytic for all finite $\ell$ but develops a first-order non-analyticity (a discontinuity in $\partial F/\partial a$) along the boundaries $a=c$ and $b=d$ as $\ell\to\infty$ (Eq.~\ref{eq:physics}; Figure~\ref{fig:panel1_conceptual}b). The four phases are thus separated by first-order transitions in the $\ell \to \infty$ limit and by analytic crossovers at finite $\ell$---precisely where the generative manifold undergoes geometric dimension reduction (Proposition~\ref{si:manifold_collapse}), the interaction matrix \bM collapsing from a 2D surface to a 1D line. We exploit this geometric collapse for model selection: if $K$ is overspecified, it introduces no spurious structure because redundant types lack the strategic differentiation required to occupy distinct manifold positions (see Results).

\paragraph{Spectral Detectability and Symmetry Breaking.} 
Whether the resulting structure is statistically detectable depends on the payoffs themselves. We characterize this detectability via the Kesten–Stigum criterion \citep{decelle2011inference}, where community structure is information-theoretically identifiable only if the signal-to-noise ratio satisfies $\mathrm{SNR} = |\lambda_2(\mathbf{C})| / \sqrt{ \langle k \rangle} > 1$.
We find that the detectable region in the payoff plane $(b,c)$ is highly non-uniform, depending heavily on the specific payoff configuration. Two distinct mechanisms shape the identifiability boundary (see SI \ref{si:detectability}). The first is a topological ceiling on the signal eigenvalue. Coordination and conflict games naturally produce block-diagonal (assortative or disassortative) network structures whose spectral gap grows with $\ell$, allowing the signal eigenvalue to saturate near the theoretical maximum $\mathrm{SNR}_{\max} = \sqrt{\langle k \rangle}$. Dilemma games are categorically different: because both types prefer the same partner, the network converges to a core-periphery structure (a rank-one, column-concentrated limit). A rank-one expected structure can support only a single detectable spectral signal, so when symmetrized, this structural overlap severely bounds the second eigenvalue, capping the maximum achievable  $\mathrm{SNR}$. Consequently, dilemma structures are information-theoretically undetectable at low mean degrees regardless of the selection intensity $\ell$. The second is a payoff-ratio nullspace. Along the locus $ad=bc$, where within-row payoff ratios balance, the partner-choice equations reduce the joint edge distribution to random-mixing. Across this strip, which threads through the intersection of the four game-type quadrants, the $\mathrm{SNR}$ vanishes entirely for all $\ell$, rendering the structure fundamentally undetectable. While these theoretical ceilings represent the worst-case limits under a parsimonious, homogeneous regular-graph assumption, degree heterogeneity and block-size asymmetry improve detectability substantially.

\begin{figure}[h!]
    \centering
    \includegraphics[width=.75\linewidth]{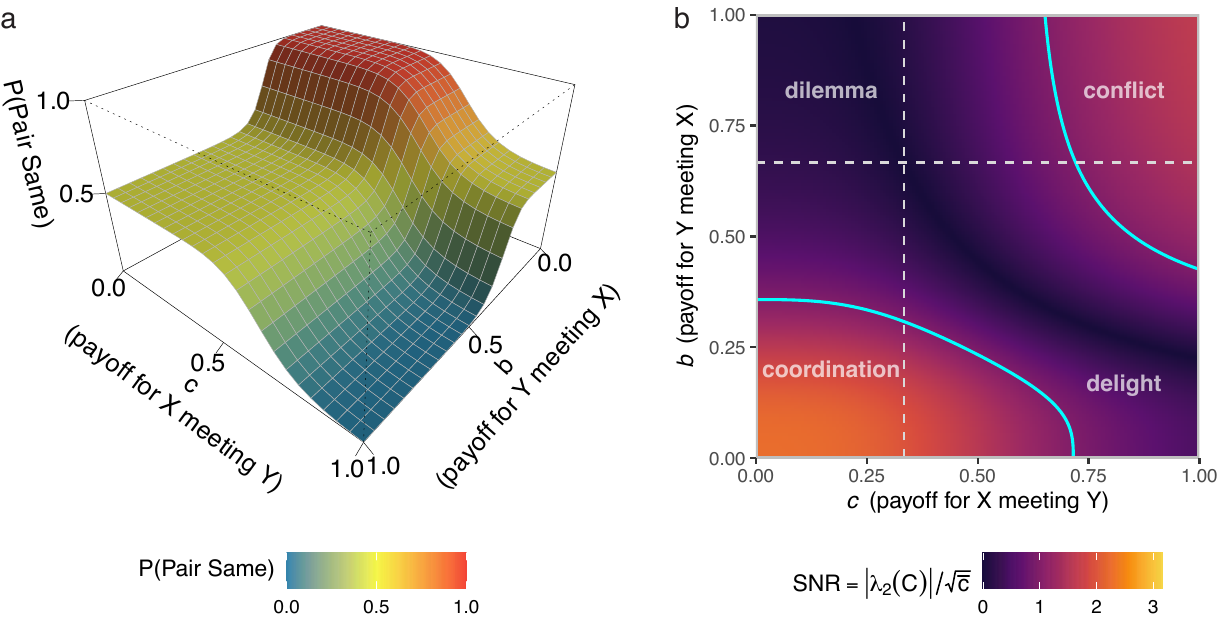}
    \caption{
    \textbf{Topological manifolds from game-theoretic phase space and theoretical distinguishability landscape.}
    \textbf{a,} Sharp crossovers separate the plateaus of the four distinct topological (Assortative, Disassortative, X-Absorptive, and Y-Absorptive) across the payoff space. The probability of assortative mixing, $P(\text{pair same})$, shifts based on independent payoff dominance thresholds ($a=c$ and $b=d$).  Crossovers sharpen into first-order transitions as $\ell \to \infty$, directly mapping microscopic strategic incentives to macroscopic mean-field topologies (figure shows $\ell = 10$).
    \textbf{b,} The heatmap of theoretical lower bound $\mathrm{SNR}$ identifiability. The landscape reveals a striking topological constraint: dilemma interactions form rank-one structures that strictly bound the signal-to-noise ratio, while a fundamental nullspace of competing payoffs perfectly balance, suppressing the latent symmetry-breaking transition and establishing an absolute physical boundary on community formation (homogeneous regular graph $x = 0.5$, $k=10$ and  $\ell = 3$).
    }
    \label{fig:panel2_physics}
\end{figure}

\section*{Results} \label{sec:results} 
We now develop an approach to infer payoffs from the fossilized record of equilibrium network topology through mechanistic inference. This allows us to test whether the mesoscale structure of an empirical network is consistent with pairwise strategic partner choice. Specifically, we formulate a stochastic block model (SBM) which we term the \oursfull (\ours) in which the block connectivity matrix \bM is not a free parameter but is instead a deterministic function of our game-graph model (see Materials and Methods for details). In a standard SBM, the probability (or rate) at which any two nodes connect depends only on their block assignments through a $K \times K$ matrix \bB whose entries are unconstrained \citep{holland1983stochastic,karrer2011stochastic}. Prior work has constrained block models by placing structured priors on \bB---for example, restricting it to ordered or layered configurations \citep{gallagher2021clarified}. Such priors require the researcher to commit to a specific structural form before fitting the SBM to data. Our approach is fundamentally different: we replace \bB entirely through the mechanistic mapping of our \FW given in Eq.~\ref{eq:physics} based on payoff matrix \bU, rationality $\ell$, and the block frequencies $\mathbf{x}$. Specifically, \bB is obtained by normalizing \bM by the number of available node pairs in each block (see Methods). This imposes a hard constraint that maps the network onto a low-dimensional manifold; it is never a free parameter. \ours then tests the precise statistical hypothesis whether the observed block structure is consistent with the restriction that \bB lies on the partner-choice manifold against the alternative that \bB is unconstrained. 
Standard model comparison like description length and Bayesian information criteria operationalize this test. By constraining inference to the game-theoretic manifold the model effectively disentangles the latent strategic causes from their stochastic topological effects. We provide both $K=2$ and $K=3$ implementations, the latter of which supports cyclical network topologies (Figure \ref{fig:3x3_extension}; Section \ref{si:gibbs}).

We fit the model using both variational expectation–maximization (EM) and fully Bayesian Gibbs sampling (MCMC). The EM algorithm provides rapid maximum a posteriori (MAP) estimation: the E-step updates soft block assignments via standard mean-field rules, while the M-step performs bounded optimization over the continuous payoff landscape (via L-BFGS-B, a gradient-based solver, \citep{byrd1995limited}). To guard against local optima, we run the EM from multiple initializations and retain the solution with the highest evidence lower bound (ELBO). However, variational EM is known to underestimate posterior uncertainty and can overfit to degenerate boundary singularities in discrete landscapes. This is the expected consequence of minimizing the exclusive direction of Kullback–Leibler divergence \citep{wolpert2006information}. To capture exact posterior volume and parameter uncertainty, we also implement a fully Bayesian Gibbs sampler, which we initialize from the EM optima. The MCMC routine uses conjugate updates for the latent block assignments and a Metropolis-Hastings step to infer the continuous payoff parameters. By integrating over the posterior, the sampling approach provides a mathematically rigorous evaluation of the model's structural regularization. Both estimators converge within minutes on a standard single-core CPU, requiring no GPU or specialized computing infrastructure.

We compare the constrained and free fits using several complementary criteria suited to both estimation methods (SI \ref{SI:model_comparison}). For the EM point estimates, we evaluate a BIC-type minimum description length ($\Sigma_{\text{BIC}}$), a Laplace-approximated description length ($\Sigma_{\text{Lap}}$), 
and the Vuong likelihood-ratio test \citep{vuong1989likelihood}. For the MCMC posterior draws, we compute fully Bayesian information criteria: the Watanabe-Akaike Information Criterion (WAIC) and Pareto Smoothed Importance Sampling Leave-One-Out Cross-Validation (PSIS-LOO-CV \citep{vehtari2017practical}), which explicitly penalize the effective complexity of the models. When the constrained model achieves a similar or lower description length or a higher expected log pointwise predictive density ($\Delta\text{ELPD}$), the additional structure imposed by the game-theoretic constraint is statistically justified: the network's topology is efficiently described by strategic interaction, and the inferred payoff matrix \bU provides a mechanistic, interpretable summary of the forces shaping connectivity.

\paragraph{From Free Parameters to Mechanistic Content.} 
Constraining \bM to the game-theoretic manifold entails four properties unavailable to unconstrained SBMs. First, the inferred communities are directly interpretable: the inferred payoffs map the observed mesoscopic topology onto the taxonomy of classical $K \times K$ games (e.g., Prisoner's Dilemma, Stag Hunt, Chicken, or Prisoner's Delight), allowing complex empirical topologies to be decomposed into their constituent evolutionary incentives. 
Second, payoff asymmetries anchor block labels to interpretable roles---provider, consumer, exploiter, producer (Table \ref{SI:roles}). These labels go beyond classical game-theory labels like ``cooperator'' and ``defector'' which conflate distinct partner-choice dynamics. Where a standard SBM is invariant under block relabeling \cite{peixoto2021revealing}, the \ours is parametric and semantically anchored. The asymmetric diagonal payoffs act as an inductive bias that fixes which type is which (Proposition~\ref{si:labeled_count}). This makes assignments comparable across restarts, estimators, and datasets, rendering ex post label-alignment unnecessary. 
Third, community membership acquires a concrete definition as \textit{strategic equivalence}: nodes share a block because they face the same payoff row, not merely because they exhibit statistically similar connectivity. This definition is ex ante (not algorithmic), semantically grounded (not just a statistical pattern), falsifiable within the detection bounds we derive, and generative. Fourth, because the model explains edges through a generative mechanism rather than free parameters, it can forecast network topology under novel population compositions, network sizes, and densities, enabling out-of-distribution generalization that no unconstrained statistical model can provide.

\paragraph{Structural Disentanglement and Heterophily.}
Traditional community detection algorithms, such as modularity maximization, typically define communities as assortative groups characterized by dense within-group and sparse between-group connectivity \citep{newman2004finding}. This assumption is often inappropriate for heterophilic or disassortative networks, where meaningful roles may be defined by patterns of interaction across groups rather than within them \citep{newman2006modularity}. Our framework sidesteps this limitation through its definition of communities via strategic equivalence. Because the interaction manifold is parameterized by game payoffs rather than assortative connectivity or structural equivalence \citep{lorrain1971structural}, the model natively captures disassortative dynamics. When the generative payoffs dictate an exploitative or bipartite dynamic (Type X survives by parasitizing Type Y; $M_{XY} \gg M_{XX}$), the model infers an asymmetric off-diagonal affinity without the mathematical penalties suffered by standard homophily-based clustering. 

\paragraph{Validation.} The detectability limits derived above are worst-case conditions assuming homogeneous degree and symmetric block sizes, both of which are relaxed in real-world networks, substantially boosting community signal. Specifically, fitting a degree-corrected SBM allows the node-specific parameters ($\theta$) to absorb the node-level degree imbalance inherent to the rank-one topology, revealing a residual community signal that the standard uncorrected spectral $\mathrm{SNR}$ fails to capture (SI \ref{SI:degree_correction}). Furthermore, \ours improves detectability through an asymmetric diagonal gauge ($a=2/3, d=1/3$), which acts as an explicit symmetry breaker, anchoring the $X \leftrightarrow Y$ label invariance and rendering the off-diagonal payoffs $(b,c)$ identifiable even in otherwise degenerate regions of the phase diagram. 
Pinning the diagonal payoff values $a>d$ removes the continuous gauge freedom but leaves an exact two-fold symmetry: swapping the community labels maps the inferred payoffs $(b,c) \to (c \cdot d/a, b \cdot a/d)$ and reproduces the identical \bM, \bB, and likelihood, so the raw orientation of any single fit is arbitrary. We break this symmetry with the convention $c/b \le a/d$: labeling $X$ the block whose within-group payoff is large relative to its cross-group payoff ($a/c \ge d/b$). The rule uses only inferred payoffs (no external covariates), is scale-free, separates the two mirror solutions  exactly at their symmetry point $c/b = a/d$, and coheres with the asymmetric diagonal pin: $X$ is then both the higher within-group-payoff type and the more homophilous one. 

We validate the \ours on synthetic networks generated from known payoff parameters. First, we plot the per-edge log-likelihood difference between the constrained and free SBMs against the $\mathrm{SNR}$ of the planted network (Figure \ref{fig:panel3_results}a and \ref{si:susceptibility}). The constrained model achieves a higher likelihood than the free model precisely when the $\mathrm{SNR}$ exceeds the Kesten–Stigum detectability threshold ($\mathrm{SNR} > 1$); below this threshold the two models are indistinguishable, confirming that the game-theoretic constraint extracts signal only where signal exists. Below the threshold, rare instances show the unconstrained model achieving a modestly higher likelihood, reflecting overfitting to residual noise rather than genuine structural signal.  This instance-level scatter is consistent with a smooth finite-size detectability transition \citep{young2017finite}. The game-theoretic constraint thus acts as a regularizer: it confines inference to the mechanistically admissible region, suppressing noise-fitting below the detectability threshold while recovering the full signal above it. Second, we find a complementary picture in parameter space (Figure \ref{fig:recovery_heatmap}; SI \ref{SI:parameter_recovery}). We plot a heatmap of the recovery error (Frobenius distance between planted and recovered \bM) over a grid of payoff values $(b, c)$. Recovery error is near zero throughout the detectable regime and spikes precisely at the Kesten–Stigum boundary. This spike is not merely a detection failure: near the phase transition, the generative manifold is maximally sensitive---tiny payoff differences straddle topologically distinct regimes, so small estimation errors in payoff space translate to large errors in recovered structure. \ours thus correctly diagnoses where recovery is inherently ill-conditioned, not just difficult.

\begin{figure}[t!]
    \centering
    \includegraphics[width=1\linewidth]{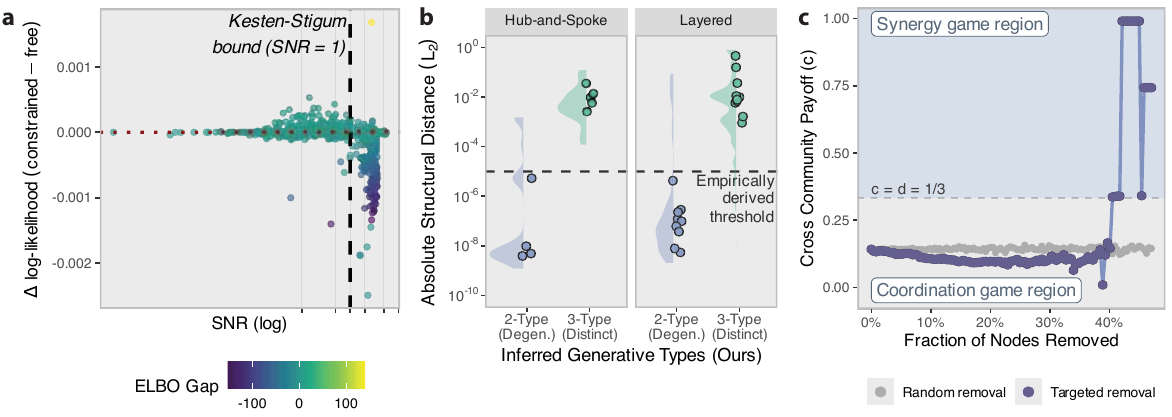}
        \caption{
        \textbf{The \ours~recovers theoretical detectability limits and reveals hidden generative structure in empirical networks.}
        \textbf{a.} Difference in log-likelihood between constrained and free models versus $\mathrm{SNR}$ for synthetic two-community block models (N = 10,000), colored by ELBO gap. The dashed line marks the Kesten–Stigum bound ($\mathrm{SNR} = 1$); the structured model is favored only above threshold. 
        \textbf{b.} Structural distance ($L_2$) of the peripheral block for empirical networks, faceted by Gallagher et al.'s descriptive ``hub-and-spoke'' vs.~``layered'' labels \citep{gallagher2021clarified}. Half-violins show densities across EM restarts; points are network-level medians. Through self-regularization, \ours~autonomously determines effective dimensionality splitting nominally homogeneous descriptive categories by their true generative structure.
        \textbf{c.} Node-removal experiment on the Polblogs network. After removing about 40\% of the highest-degree nodes, the inferred payoffs flip from a coordination game to a synergy game (crossing the $c=d$ boundary); random node removal does not systematically affect the inferred payoffs.
        }
    \label{fig:panel3_results}
\end{figure}

\paragraph{Empirical Analysis.}
To validate the generative constraint, we apply the \ours to canonical empirical datasets across social, biological, and economic domains (SI Table \ref{tab:datasets}). On networks including Zachary's Karate Club, Dolphins, and H.\ sapiens protein-protein interactions (PPI), the model compresses topological structure into a two-parameter game-theoretic manifold that decisively outperforms the unconstrained SBM on description length ($\Delta\Sigma_{\mathrm{Lap}} = -7.1$ to $-27.5$ nats; SI Table \ref{tab:empirical_results}). For Zachary's Karate Club, description-length comparison yields $\Delta\Sigma_{\mathrm{Lap}} = -7.1$ nats in favor of the constrained model (strong evidence on the Kass–Raftery scale; \citep{kass1995bayes}): a two-parameter delight game explains the network better than a free three-parameter $\mathbf{B}$ matrix. The inferred parameters yield immediate structural interpretation: for Zachary, the model infers the two ethnographically documented factions \citep{zachary1977information} and identifies a diagonal-dominant (assortative) regime---both groups preferentially form within-group ties (up to the $\ell$ normalization; see Methods). The inferred cross-group payoffs further reveal that the administrator bloc gains more from interacting with Mr.~ Hi's faction than vice versa, consistent with the sociological account of an informal, relationship-driven network fracturing away from a weaker institutional structure.
When networks do not conform to partner-choice dynamics, \ours~honestly reports null results without hallucinating apparent community structure. On C.~elegans and yeast PPIs the optimizer finds undifferentiated payoffs ($b = d$ and $c = a$): agents are indifferent to connecting within cluster or across, indicating random mixing, giving us a formal ER diagnostic.\footnote{This is also diagnosable from the undefined $\Delta\Sigma_{\mathrm{Lap}}$ caused by the overlapping mixture components which causes the Fisher Information matrix to lose rank.} 
The game-theoretic manifold captures the disassortative pattern but cannot fully accommodate the observed mixing proportions, suggesting additional biological mechanisms (e.g., spatial co-localization, complex-level constraints) beyond pairwise strategic logic as predicted in the biology literature \citep{gavin2002functional}.

We validate the structural integrity of the game-theoretic manifold using the Bayesian MCMC sampler (Table \ref{tab:results_gibbs}). Using international trade network data, \ours infers payoffs of a directed core-periphery with an asymmetric extraction game dynamic (i.e., an unequal exchange game; Dilemma), in which core nodes extract surplus payoffs from trading with peripheral nodes ($b > a$), and peripheral nodes are penalized for trading with the core compared to trade within the periphery ($c < d$). Core-periphery topology alone cannot differentiate between an extractive and a public good relationship. Our model mathematically infers directional exploitation from its topological structure. We find near identical payoffs and block frequencies between the Bayesian and the EM estimates, suggesting the EM point estimates represent the true, stable posterior mass. 
Out-of-sample predictive accuracy criteria (WAIC and PSIS-LOO-CV) confirm that the game-theoretic constraint achieves predictive accuracy equivalent to the unconstrained baseline with fewer effective parameters. We demonstrate expressivity of our model to capture intransitive, non-assortative networks with a cyclic-dominance game with three groups and several other positive and negative cases (Section \ref{si:gibbs}).

\paragraph{Structural Forecast Test.}
Predicting the evolution of complex networks under shifting demographic constraints remains a fundamental challenge in machine learning and network science \citep{goldenberg2010survey}. Standard statistical models, such as SBMs, ERGMs, or Graph Neural Networks, describe static topologies well but fail to predict how edge probabilities rewire when the relative sizes of node classes change. Because \ours~contains the mechanistic engine of network formation, it can do what standard models cannot: zero-shot forecasting under demographic shifts (i.e., out-of-distribution generalization). Constraining the machine learning model with the physics of evolutionary game theory enables inferences about the forces that shape trade such as the capacity constraints and crowding effects that make strategic position costly to change, and the frequency-dependent partner preference that governs interaction given that position. Notice also that our forecasts draw only on network-derived quantities and require no external covariates. A gravity model, by contrast, would require economic masses (GDP) and bilateral distances.

We demonstrate this using 31 annual global trade networks (1994-2024) and perform a temporal prediction experiment that illustrates this capability. Specifically, we take the inferred payoff parameters from a network at time $t$ ($b_t, c_t, x_t$) and feed them into \ours~along with macroscopic population proportions $x$ at time $t+1$. Macroscopic population proportions can be specified either exogenously (e.g., data about trade bloc accession, sanctions, or decoupling) or endogenously from population dynamics (e.g., best-response dynamics, cultural evolution, replicator dynamics, or the Moran process).\footnote{We proxy $x_{t+1}$ using \ours~applied to the network at $t+1$. Note that just knowing node proportions contains zero information about actual network structure and hence does not invalidate the test.} The model then predicts macroscopic network ties ($\mathbf{M}$) as endogenously generated by evolutionary partner-choice (i.e., zero-shot prediction of out-of-sample network rewiring conditional on that new state, generalizing to out of distribution). The model can also predict counterfactual network ties under changing incentives (e.g., tariffs). We compare this to two statistical baselines (random mixing or assortativity inertia) across $t \rightarrow t + \{1, 2, 3\}$ year-pairs using Mean Absolute Error (MAE) between $\hat{\mathbf{M}}_P - \mathbf{M}_{t+1}$ (Table \ref{tab:trade_forecasting}). Forecasting one year ahead, we find substantial +98.7\% improvement over the random mixing baseline (Wilcoxon signed rank $p = 9.3 \times 10^{-10}$) and +4.8\% over the inertia baseline ($p = 0.0249$). We find similar improvements predicting two or three years ahead. Replacing exogenous population proportions with endogenous $x_{t+1}$ predictions from replicator dynamics yields predicted directions of $\Delta\mathbf{M}$ that are uncorrelated with observed change ($p = 0.66$). This aligns with the central premise of our framework: changing strategies is costly and constrained by institutional frictions and capacity constraints (i.e., countries cannot simply decide to be major trading partners and join the network's core), but conditional on one's strategic position, choosing interaction partners follows game-theoretic dynamics. The success of these out-of-distribution forecasts is evidence that global trade networks are not scaling mechanically with economic growth but are actively rewiring according to game-theoretic partner choice.

\paragraph{Self-Regularization: Automatic Block Selection.}
Our framework naturally performs self-regularization through the geometry of its parameter space: when $K$ is over-specified, the partner-choice manifold has no room for a redundant strategy type (Proposition~\ref{si:manifold_collapse}). The model selects its own effective $K$ as a byproduct of fitting, without relying on description length or cross-validation over multiple $K$ values. Because of the geometric manifold collapse, the constrained model cannot invent structural patterns to accommodate higher-order structure. We demonstrate that the analytically derived regularization achieves implicit model selection in practice through this parameter degeneracy as follows. We pre-specify the maximum interaction order ($K=3$) in \ours and pin the diagonal payoff parameters of two peripheral blocks to be identical---thus imposing equivalent internal connection densities $M_{22}$ and $M_{33}$. We illustrate the topological manifold collapse by measuring the absolute $L_2$ distance between $\mathbf{B}[2,]$ and $\mathbf{B}[3,]$. Note that the $L_2$ threshold is only a diagnostic quantity used to visualize the degree of manifold collapse, not a decision rule for model selection per se. Model selection operates implicitly via multiple restarts, selecting the lowest (collapsed or non-collapsed) solution as the data dictates. Previous model selection approaches force researchers to explicitly enumerate and independently fit model classes, such as testing a strict two-block hub-and-spoke model against a layered model across varying discrete depths. Instead, our framework embeds the network within a single, continuous generative parameter space. This manifold collapse creates a flat optimization ridge. Across restarts of the EM fitting procedure, the optimization trajectories naturally bifurcate and the true dimensionality spontaneously resolves via basin selection. The network either collapses into a degenerate $K=2$ manifold---indicating the data only support a $K=2$ game---or sustains a structurally distinct $K=3$ topology. In practical terms this results in two clusters with identical payoffs and non-zero soft weights, but one is empty on hard weights. 

We apply this framework to the 28 empirical communication networks classified as ``hub-and-spoke'' vs.~``layered'' by \citep{gallagher2021clarified} and reclassify them (Figure \ref{fig:panel3_results}b). We pick the degenerating threshold empirically as $L_2 < 10^{-5}$ by analyzing the bimodal pattern of EM restarts (Figure \ref{fig:implicit_model_selection}).
We find both $K=2$ and $K=3$ dimensions exist across the reported classification: 5/9 reported hub-spoke networks show distinct $K=3$ structure and 9/19 reported layered networks collapse to $K=2$.\footnote{We test only $K=3$, but this is sufficient: Peripheral collapse indicates that the data does not support the simplest layered structure of two distinct peripheral roles that Gallagher's typology would predict.}
The descriptive classifications from Gallagher et al.~conflate fundamentally different generative systems. Because traditional taxonomies test networks against rigid structural templates (e.g., nested $k$-cores), they often force topologically diverse 3-type networks into 2-type buckets if they fail to match a specific heuristic. By contrast, \ours~cleanly separates networks driven by only two behavioral types (degenerate $L_2 < 10^{-5}$) from those requiring three distinct strategic types, regardless of whether those interactions form traditional ``layers''.\footnote{We confirmed such degeneracy structure using the full Bayesian (Gibbs) estimator: $K=3$ chains exhibit catastrophic mixing failure (high R-hat, very low ESS) precisely on the networks where EM detects collapse, while $K=2$ chains mix cleanly (R-hat of 1.0 and high ESS). The same geometric degeneracy that EM resolves via basin selection breaks the MCMC posterior sampling.}

\paragraph{Case Study: Deplatforming and the Persistence of Echo Chambers.} 
Intervening on a network's topology does not inherently alter the incentives that generate it (Figure \ref{fig:panel3_results}c). We use the \ours as a policy diagnostic on the political-blog network \citep{adamic2005political}. We treat the fraction of deplatformed accounts as a control parameter and then re-infer the underlying game at each step (details in \ref{si:deplatform}). The system is governed by a coordination game: blogs are rewarded for linking within their own political camp. Under standard topological assumptions, removing the most central hubs should disrupt this structure \citep{albert2000error}. Instead, the \ours reveals that the coordination regime exhibits strong behavioral resilience. Targeted removal of the highest-degree accounts leaves the echo-chamber incentives intact until roughly 40\% of the network is gone. Only then do the inferred payoffs cross the phase boundary to the synergy regime with a mainstream vs.~niche interaction structure. By contrast, indiscriminate removal induces no systematic change in interaction payoff. It fails to systematically force the network off its echo-chamber equilibrium. Thus, deplatforming changes the players but not their roles, and the structure simply reconstitutes around new hubs until the intervention is large enough to force a qualitative shift in behavioral incentives. By mapping this phase boundary, the \ours  tells a practitioner exactly how much targeted disruption is required to overcome a network's behavioral resilience. Ultimately, this resolves a critical gap in intervention design: while foundational theory establishes that complex networks are structurally resilient to random failures but vulnerable to targeted attacks \citep{albert2000error}, \ours reveals they possess a deeper behavioral resilience in strategic equivalence, rendering even targeted structural attacks ineffective unless they simultaneously cross the system's game-theoretic phase boundary.

\section*{Discussion}
Complex networks have long been characterized by what they look like---small-world \citep{watts1998collective}, scale-free \citep{barabasi1999emergence}, modular \citep{girvan2002community}, hierarchical \citep{ravasz2003hierarchical}---but rarely by why those structures emerge. Our work shifts that frame by casting network topology as the fossilized record of strategic interaction, the result of an environment in which changing partners is (relatively) easy but changing strategies is hard. We demonstrate that the infinite, continuous parameter space of all possible two-player strategic interactions collapses strictly into discrete macroscopic regimes, for which we can draw precise phase boundaries and quantify detectability. In the $2 \times 2$ case, the resulting four topologies---assortative, disassortative, and two core-peripheries---that network science has treated as descriptive, are in fact the exhaustive thermodynamic phases of an evolutionary game. Comparing payoffs along one dimension ($a=b$, $c=d$) separates game-theoretic Nash equilibria and comparing across another separates equilibrium topologies ($a=c$, $b=d$; Figure \ref{fig:panel2_physics}a-b). A single Nash equilibrium can generate several distinct topologies, and conversely a single topology can be realized by distinct games. Networks that are structurally identical under a standard block model can arise from completely different causal mechanisms: exploitation, mutual attraction, or coordination pressure are all consistent with the same core-periphery block matrix, but immediately distinguishable in the inferred payoffs. 
The reason is structural: the diagonal anchor supplies type-specific information that relational data lacks, resolving all $K^K$ labeled topologies natively, whereas an exchangeable SBM sees only permutation classes, breaking the cross-dataset comparability that label-matching schemes must then reconstruct after inference \citep{peixoto2021revealing}.

This carries an unexpected implication for game theory. The dilemma region of the payoff space is by far the most studied in the literature \citep{axelrod1981evolution,santos2005scale}, yet it produces only a single, fully determined network topology under partner-selection dynamics as pairwise incentives overdetermine the absorptive outcome. 
The delight region, by contrast, has received comparatively little theoretical attention. Although such games are simple in that they are dominance-solvable and pose no social dilemmas, they can generate rich structural diversity as all four topological networks are seen under the same Nash equilibrium.

Conceptually, the \ours introduces a coarse-graining \citep{gabrielli2025network} of individuals into interaction types. This offers a clear definition of communities as \textit{strategic equivalence} groups---nodes share a block strictly because they face the same payoff row. This provides the first mechanistically grounded definition of community structure, resolving a foundational ambiguity in community detection \citep{fortunato2010community,peixoto2021revealing,ahn2010link}. Traditional definitions instead rely either on heuristic assumptions of homophily and edge density (e.g., modularity maximization \citep{girvan2002community}) or on purely statistical regularities such as identical connection probabilities \citep{peixoto2019bayesian}. 
Crucially, this coarse-graining is not merely consistent across scales but generative. The strategic types infer the evolutionary or economic mechanism that produced the fine-grained network to begin with. 
This also supplies a mechanistic origin for the low-rank hypothesis of complex systems \citep{thibeault2024lowrank}: the rapidly decaying singular values observed empirically across real networks are the spectral shadow of strategic equivalence, with the effective rank set by the number of strategically distinct types and dropping by one each time two types collapse onto a shared payoff row.
Dynamically, a GGE is a mean-field synchronized fixed point where same-type nodes converge to a shared link-strategy through the population field, not through mutual coupling. This is the mean-field product distribution of Product Distribution theory \citep{wolpert2006information}, but exact rather than approximate: strategic equivalence makes same-type agents face identical incentives, supplying at the type level the decoupling that mean-field theory must otherwise assume.
By grounding community structure in a generative mechanism rather than unconstrained statistical parameters or homophilic heuristics, this framework naturally captures both assortative and heterophilic dynamics, supplying the explanatory foundation that prior models lack. 
This also dissolves a long-standing puzzle of why some communities are stable while others vanish \citep{palla2007quantifying}: what persists is not a roster but a payoff row. 

The abstract nature of our framework permits multiple interpretations. The game-graph dynamics describe associations as changing due to interaction payoffs, but this can happen in ways that are not a product of partner choice. For example, our model could arguably apply to microbial communities where associations may be driven by reproductive success rather than individual learning. A revealing case is described by Momeni et al.~\citep{momeni2013elife}, who show that in spatially growing two-population microbial communities, strong bidirectional cooperation uniquely produces a characteristic intermixed spatial pattern, while competition produces segregation and commensalism produces layering. In our model, obligate microbial mutualisms have a payoff matrix of the form $a \ll c$, $b \gg d$ where each type grows faster near complementary partners than near its own type. This places the system in the Disassortative quadrant of the phase diagram, predicting bipartite interaction networks analogous to spatial intermixing. By contrast, competition between strains that share no metabolites places the system near the origin of the phase diagram ($a \approx c$, $b \approx d$). In this case, no type has a payoff-based reason to drive the interaction structure and segregation reflects only kin clustering from local reproduction. 

The phase bounds in the game-to-network mapping have implications for network interventions. When designing metabolic interactions in biotech applications or policy changes that alter interaction payoffs---tariffs, sanctions, subsidies---will only qualitatively restructure a network if they push payoffs past a phase boundary. Phase boundaries can also help identify high-leverage targets (e.g., in protein-protein-interaction) where small perturbations would reorganize interaction patterns, acting as computational target prioritization tool. More broadly, the \ours provides a diagnostic for how collective behavior is organized in multi-agent systems: fitting the model to an observed interaction topology infers whether the underlying dynamic is one of coordination, conflict, or exploitation, without inspecting individual agent behaviors.
The same logic bears on contested social structures: homophilous, polarized clusters (e.g., so-called echo chambers) correspond to the assortative regime, so interventions reshape them only if they move payoffs across a phase boundary rather than merely adjusting incentives within it.
Our framework also serves as a diagnostic for interventions one does not control: using $\ell$ an analyst could ask whether an observed echo chamber reflects curation amplifying a pre-existing homophilous substrate or a genuine move of the payoffs across a boundary.





The $K \times K$ generalization deepens the mechanistic account of our framework. The four-topology result for $2 \times 2$ games is a provable consequence of binary preference decisions, and extending this to $K$ types unlocks qualitatively new structures such as cyclic dominance (e.g., rock-paper-scissors dynamics) which is structurally impossible when agents choose between only two actions. Crucially, this richer expressiveness does not come at the cost of overfitting: when no additional strategic type is warranted, the game-theoretic manifold provably has no room for it, collapsing redundant payoff rows and self-regularizing in the same sense that LASSO does---as a population-level geometric guarantee, with finite-sample estimators inheriting the usual ridge-wandering caveats. No description length metric is necessary to make that decision.

As a result, in the \ours, the right $K$ has a precise definition: the number of strategically distinct types, where ``distinct'' means having a different ordinal best-response. This is theoretically grounded in the number of qualitatively different strategic roles that exist in a system, is falsifiable (within theoretically defined detection bounds), immune to proliferation from random noise because random interactions have no distinct strategic identity, and each community is semantically interpretable. Our empirical application of the \ours to 28 empirical communication networks reveals existing descriptive taxonomies conflate networks with fundamentally different generative dimensionality, suggesting that canonical typologies across the network literature may systematically misattribute mechanism, and that auditing them through \ours would separate genuine behavioral universals from descriptive coarse-graining. 

Identifying whether a network's mesoscale structure is reducible to pairwise interactions is the central open question of the higher-order networks program \citep{battiston2020networks,battiston2021physics,young2021hypergraph}. By fitting an explicitly pairwise generative model, \ours converts that question into a falsifiable diagnostic with four distinct verdicts. 
First, if the constrained model fits well with inferred payoffs indicating a unique equilibrium, the network is fully explained by pairwise strategic logic, ruling out higher-order structure as the most parsimonious explanation. 
Here, $\ell$ provides an additional diagnostic: small $\ell$ indicates purely pairwise interactions while large $\ell$ produces path-dependent preferences in \bM that could be mistaken for higher-order dependencies. What appears to be simultaneous group interaction is exposed by large $\ell$ as sequential pairwise interaction with memory, the structural counterpart of unrolling synergistic integration into sequential steps \citep{riedl2026emergence}.
Second, a clean fit in a multiple-equilibrium region indicates that pairwise logic explains the baseline mechanism, but equilibrium selection requires something else. This directly points to higher-order dependencies as the explanatory mechanism.\footnote{Given multiple replications of the data, \ours could discriminate between random equilibrium selection and systematic selection based on beyond-pairwise mechanisms.}
Third, if the constrained model infers valid interior payoffs but the free model fits better, this points to a strong higher-order candidate. A superior fit by the free model implies the network's mesoscale organization is driven by complex group-level rules or multi-agent dependencies that cannot be compressed into independent dyadic evaluations \citep{riedl2026emergence,riedl2026claims}.
Fourth, if \ours infers undifferentiated payoffs, this is a clear signal that no pairwise strategic signal is present at all. This is a necessary but not sufficient condition for a truly random network, or it could point to a strict higher-order structure that enforces pairwise uniformity (like XOR) and remains invisible to dyadic projections. The two cases are distinguishable by information decomposition: XOR-type structure is purely synergistic information about the topology, which is zero in every dyadic projection but detectable from the joint distribution \citep{riedl2026claims}. These modes assume single-type membership; pervasive overlap \citep{ahn2010link}, where each node belongs to many communities, is a separate source of pairwise misfit not reducible to higher-order interaction.
Together, these four diagnostic modes define a causal horizon of what pairwise strategic interaction can and cannot explain. 


Our approach shares deep theoretical roots with the Hidden Metric Space framework \cite{krioukov2010hyperbolic}. In geometric graph theory, topological structure emerges from distances in a latent hyperbolic manifold, while in our work ``distance'' is payoff-based, and $\ell$ plays the same role as inverse temperature in both frameworks. The analogy extends to renormalization. Geometric renormalization coarse-grains a network by aggregating nodes nearby in the latent space and rescaling the embedding so that the connection probability is left invariant \citep{gabrielli2025network}; the \ours aggregates nodes that are not merely nearby but \textit{identical} in payoff space, with the same connection-probability form preserved under the partner-choice map. The degree correction ($\theta$) separates the ``popularity'' dimension from the ``strategy'' dimension, mirroring the radial/angular decomposition in hyperbolic models \citep{papadopoulos2012popularity}, or the leading singular component of real networks \citep{thibeault2024lowrank}. A natural extension would combine both mechanisms, allowing edges to be driven jointly by spatial proximity and incentives. This could be particularly relevant for protein interaction networks where binding depends on both co-localization and functional complementarity.

Our partner-choice equations extend naturally to $k$-player games, where the mixing matrix generalizes to a mixing tensor over group compositions, yielding a principled test for whether observed higher-order structure reflects genuine multi-body strategic logic or is reducible to pairwise interactions \citep{battiston2020networks}.


\section*{Materials and Methods}
\subsection*{Model}
\label{methods:model}

We model pairwise interactions in a $K$-type population as a $K \times K$ symmetric game. Let $U_{ij}$ represent the payoff to strategy $i$ from interacting with a type-$j$ partner. A population state is a vector $x = (x_1, \ldots, x_K)$ where $x_i$ is the frequency of type $i$. Let $q_{ij}$ denote the probability that a type-$i$ individual directs an outgoing link to type $j$, with $\sum_j q_{ij} = 1$ for each $i$.

A \textit{Game-Graph Equilibrium} (GGE) is a link-probability profile $(q_{ij})$ such that for all $i$,
\begin{equation}
    q_{ij} > 0 \quad \text{only if} \quad
    U_{ij} \geq U_{ik} \quad \text{for all } k.
    \label{eq:game-graph-equilibrium}
\end{equation}

We formulate network dynamics as driven by an affinity between types based on the game payoff $U$ and a background payoff $w$ (ensuring positivity):
\begin{equation}
\alpha_{ij} = U_{ij} + w.
    \label{eq:affinity}
\end{equation}

The \textit{Game-Graph Dynamics} are given by
\begin{equation}
q'_{ij} = q_{ij} \cdot \frac{\alpha_{ij}}{\sum_k q_{ik} \alpha_{ik}}.
\label{eq:graph-dynamics}
\end{equation}
This is a replicator dynamic applied to link probabilities: $q_{ij}$ grows when affinity for type $j$ exceeds the current average, and shrinks otherwise. We assume interactions begin with random mixing and the dynamics operate for $\ell$ iterations. GGEs are rest-points of this dynamic ($\ell \to \infty$), and every asymptotically stable point is a GGE. The same is true of GGEs for any monotone selection dynamic (SI Section~\ref{si:analytical}, Corollary Case~1). This generalizes to hold for $K>2$ (SI Section~\ref{si:analytical}.

The joint interaction frequency $M_{ij} = x_i\,q_{ij}$ is the quantity inferred from network data and used in the SBM formulation. For $\ell$ learning generations, in $2 \times 2$ games, the resulting joint probabilities of type pairs are: 
\begin{equation}
\begin{aligned}
    M_{XX} &= \frac{x^2\,a^{\ell}}{x\,a^{\ell} + y\,c^{\ell}},
    &\quad
    M_{XY} &= \frac{x y\,c^{\ell}}{x\,a^{\ell} + y\,c^{\ell}}, \\[6pt]
    M_{YX} &= \frac{x y\,b^{\ell}}{x\,b^{\ell} + y\,d^{\ell}},
    &\quad
    M_{YY} &= \frac{y^2\,d^{\ell}}{x\,b^{\ell} + y\,d^{\ell}}.
\end{aligned}
\label{eq:physics}
\end{equation}
where $x$ and $y$ are the frequencies of strategy $X$ and $Y$ respectively and $a,b,c,d$ are the payoffs of the game (Fig.~\ref{fig:panel1_conceptual}; derivation in SI Section~\ref{si:analytical}).

\subsection*{Stochastic Block Model}
We consider an unweighted, directed, simple network with $N$ nodes and adjacency matrix $\mathbf{A}$. Each node $i$ is assigned to one of $K$ blocks, recorded in the vector $\boldsymbol{x}$ where $x_i = r$ denotes membership in block $r$. Block sizes are $\mathbf{n} = (n_1, \ldots, n_K)$ with corresponding frequencies $x_r = n_r / N$. These block proportions are identified with the population-state vector $x = (x_1, \ldots, x_K)$ of the partner-choice model: nodes in block $r$ are agents of strategic type $r$. Connection probabilities between blocks are specified by an asymmetric $K \times K$ matrix \bB, so that the likelihood factorizes as
\[
P(\mathbf{A} \mid \boldsymbol{x}, \mathbf{B}) = \prod_{r,s} B_{rs}^{\,m_{rs}}\,(1 - B_{rs})^{\bar{m}_{rs} - m_{rs}},
\]
where $m_{rs}$ is the number of observed directed edges from block $r$ to block $s$, and $\bar{m}_{rs} = n_r n_s$ is the maximum number of possible directed edges (or $n_r(n_r - 1)$ for $r = s$). In the unconstrained (free) SBM, each entry of \bB is estimated independently.

\subsection*{Game-Theoretic Constraint}
Unlike prior-based approaches that restrict the block matrix to an allowed region of parameter space we impose a hard constraint: \bB is not inferred directly but is computed deterministically from the partner-choice physics in Eq.~\ref{eq:physics}. The diagonal entries of \bU are fixed by a gauge condition. To fix the residual community-label orientation we adopt the convention
$c/b \le a/d$: the block labeled $X$ is the one whose within-group payoff, relative to its cross-group payoff, is the larger of the two ($a/c \ge d/b$). The $K(K{-}1)$ off-diagonal entries are the only free continuous parameters of the model. Because the off-diagonal payoffs are unconstrained within their bounds, the manifold of reachable \bB matrices spans all connectivity patterns that the $K$-type game can generate: the model discovers the structure rather than having it prescribed.

Given the payoffs \bU, block frequencies $\mathbf{x}$, and rationality parameter (inverse temperature) $\ell > 0$, the joint edge-type distribution is given by
\[
M_{rs} = \frac{x_r\, x_s\, U_{rs}^{\,\ell}}{Z_r}, \qquad Z_r = \sum_{s=1}^{K} x_s\, U_{rs}^{\,\ell}.
\]
Given the softmax selection rule in Eq.~\ref{eq:physics}, $\sum_{rs} M_{rs} = 1$ by construction. 

Because the partner-choice dynamics are inherently asymmetric, \bM acts as a directed, bidirectional encoder of expected edge fractions. We obtain the per-dyad generative parameters \bB by distributing these global fractions across the available node pairs, utilizing a link function appropriate for the network's topology and likelihood.

For \textbf{directed networks}, the mapping is strictly asymmetric:
\[
B_{rs} = \frac{|E|\, M_{rs}}{n_r n_s}.
\]
For \textbf{undirected networks}, the asymmetric flows must be combined. Under a Degree-Corrected (Poisson) SBM, $\mathbf{B}$ represents a base expected rate. Because the sum of two independent Poisson processes is a Poisson process, we use the linear sum:
\[
B_{rr} = \frac{2\,|E|\, M_{rr}}{n_r(n_r - 1)}, \qquad B_{rs} = \frac{|E|\,(M_{rs} + M_{sr})}{n_r\, n_s} \quad (r \ne s).
\]
Conversely, under a standard (Bernoulli) SBM, \bB represents a strictly bounded probability. To prevent boundary violations when combining dense flows, we use the probabilistic ``OR'' link function, asserting an undirected edge forms if either $r$ chooses $s$ or $s$ chooses $r$:
\[
B_{rs} = 1 - \left(1 - \frac{|E|\, M_{rs}}{n_r n_s}\right)\left(1 - \frac{|E|\,M_{sr}}{n_r\, n_s}\right).
\]
where $|E|$ is the total number of edges. The composite mapping $\bU \mapsto \bM \mapsto \bB$ is deterministic throughout: each setting of the payoff parameters \bU (and optionally $\ell$) uniquely determines the joint edge-type distribution \bM via the partner-choice equations, from which \bB is then computed through the appropriate link function. The adjacency matrix \bA is then sampled from \bB. This reduces the free parameters from $K(K{+}1)/2$ in the unconstrained SBM to at most $K(K{-}1)$ and ensures that every fitted \bB admits a game-theoretic interpretation. Mathematically, \bB acts as the kernel of a step-function graphon \citep{lovasz2012large}, encoding the generative rules of the network in a representation that is invariant to discrete finite size. Because these rules are parameterized by payoffs rather than raw edge densities, the same $\phi$ generates the correct \bB at any $N$ and any block composition vector $x$---a scale-free parameterization that enables reliable out-of-distribution forecasting of network topologies.

\subsection*{Estimation and Model Comparison}
We implement estimation procedures for both variational Expectation-Maximization (EM) and fully Bayesian estimation via MCMC (SI \ref{si:EM} and \ref{si:gibbs}). We compare the constrained and free SBMs using multiple criteria, evaluated from converged fits using hard class assignments $\hat{\boldsymbol{x}}$ (obtained by $\arg\max_r q_{ir}$ for variational EM) and full posterior draws (via MCMC sampling): Description Length (BIC and Laplace), Vuong Test, and Bayesian out-of-sample prediction (WAIC and PSIS-LOO-CV; SI \ref{SI:model_comparison}).

\paragraph{Parameter Counts and Identification.}
For $K$ blocks, the constrained model has $K(K{-}1)$ free parameters (the off-diagonal payoff entries). For undirected networks, the free SBM has $K(K{+}1)/2$ parameters (the upper triangle of $\mathbf{B}$); for directed networks, the free SBM has $K^2$ parameters. Both models additionally share the partition encoding cost $L(\hat{\boldsymbol{x}})$. The degree-corrected variants add $N$ propensity parameters to each model. 
The inverse temperature $\ell$ is fixed rather than inferred to break scaling degeneracy and ensure the Fisher information matrix is positive-definite. Fixing $\ell$ is standard scale normalization \citep{train2009discrete}, not a substantive assumption: as in any quantal-response model \citep{mckelvey1995quantal}, $\ell$ and the payoff magnitude are jointly unidentified from a single network, because the likelihood depends on the payoffs only through the $\ell$-powered ratios $(a/c)^{\ell}$ and $(d/b)^{\ell}$. The direction of the payoff vector is therefore identified independently of $\ell$---the signs of the within-type comparisons, which fix the equilibrium phase, together with ratios such as $\ln(a/c)/\ln(d/b)$---while its magnitude is not; inferred payoff levels and any cross-network comparison of magnitudes are accordingly conditional on the chosen $\ell$, whereas the phase classification is gauge-invariant.

\subsection*{Synthetic Experiments}
Both synthetic validations use $K = 2$ blocks with the rationality parameter fixed at $\ell = 3$, network size $N = 500$, and target mean degree $\langle k \rangle = 10$ (setting $|E| = 2{,}500$). Graphs are sampled from the planted SBM defined by the partner-choice equations, retaining only instances where the largest connected component contains at least 90\% of nodes. Each fit uses 10 random restarts.

\section*{Acknowledgments}
We thank Albert-L\'aszl\'o Barab\'asi, Erik Weiss, and Ahmad Ahmadi for helpful comments.

\bibliographystyle{unsrtnat}
\bibliography{references}

\appendix
\input{SI}

\end{document}

%% file: SI.tex
\clearpage
\setcounter{page}{1}
\setcounter{figure}{0}
\setcounter{table}{0}
\setcounter{equation}{0}
\setcounter{section}{0}
\renewcommand{\thepage}{S\arabic{page}}
\renewcommand{\thefigure}{S\arabic{figure}}
\renewcommand{\thetable}{S\arabic{table}}
\renewcommand{\theequation}{S\arabic{equation}}
\renewcommand{\thesection}{S\arabic{section}}

\addcontentsline{toc}{section}{Supplementary Information} 
\part{Supplementary Information} 
\parttoc 


\clearpage
\section{Relationship to Prior Frameworks} \label{si:related_work}

Our work connects in particular with a large literature on Stochastic Block Models (e.g., \citep{karrer2011stochastic,peixoto2017nonparametric,young2018universality}). Peixoto \citep{peixoto2023descriptive} establishes that Bayesian inference of the stochastic block model provides statistically principled community detection, but the inferred block matrix remains a statistical summary of edge densities---it answers \textit{whether} mesoscale structure is warranted, not \textit{what mechanism} produced it. We close this gap by deriving the block matrix analytically from a frequency-dependent partner-choice process, giving the block matrix a causal interpretation for the first time: each entry is the equilibrium probability that an agent of one type selects a partner of another type, determined by payoffs, population composition, and selection intensity. While standard SBMs offer a generative model of the network via its block matrix, we provide a semantically interpretable generative mechanism of the block structure itself. That is, we replace the SBM's unstructured prior over the block matrix with a deterministic derivation from interpretable game-theoretic parameters. Furthermore, temporal SBMs (TSBM) need the temporal contrast across multiple time periods \citep{chanpuriya2022direct} while our generative model can predict zero-shot from a single observed network. 

Our work complements the broader literature applying game theoretic concepts to individual-level strategic link formation  \citep{jackson1996strategic,bala2000noncooperative,overgoor2019choosing,kuleshov2015inverse,goyal2005network,jackson2002evolution,skyrms2000dynamic}. The connection of our paper with \citep{jackson1996strategic} is particularly relevant. In that work, link formation is governed by a value function $v(g)$ which assigns a value to each possible graph, an allocation rule distributes that worth to players, and a network is pairwise stable if no player gains by severing a link and no pair gains by adding one. That formulation is extremely general: $v(g)$ can in principle encode any payoff structure, including game-theoretic interactions. However, its generality means it cannot predict which class of network a given incentive structure will produce; it can only evaluate whether a specific proposed network is stable against single-link deviations. Our framework differs in three respects that are not merely notational but are what enable our analytical results: agents choose interaction partners rather than links, with the network emerging as a population-level equilibrium of frequency-dependent selection; the equilibrium concept is a rest-point condition for the full class of monotone selection dynamics rather than a Nash-style deviation condition on individual links; and payoffs are local primitives of pairwise interaction rather than properties of the global graph. These differences are precisely what yield the topology classification, phase boundaries, detectability bounds, constrained inference via the \ours, and out-of-distribution forecasting---results that require the population-level equilibrium equations (Eqs. 4–7) and have no analogue in individual-level link-deviation logic. The two frameworks are thus complementary: \citep{jackson1996strategic}'s pairwise stability answers ``is this particular network stable?''; our game-graph equilibrium answers ``what kind of network do these payoffs produce?''

Our framework is also related to Product Distribution theory \citep{wolpert2006information}, which also uses the maxent/Boltzmann structure to bridge bounded-rational game theory and mean-field statistical physics.

\clearpage
\section{Analytical Results}

\input{SI_analytical_results}

\clearpage
\input{SI_estimator_EM}

\input{SI_model_comparison}

\clearpage
\input{SI_detectability.tex}

\clearpage
\section{Game Taxonomy and Role Labels}
The \FW provides distinctions among payoffs in the Delight quadrant as well as distinctions among games leading to core-periphery networks that offer nuanced interpretation of roles given the semantics of their payoffs (Table \ref{SI:roles}).

\input{SI_game_taxonomy_table.tex}

\section{Susceptibility Analysis} \label{si:susceptibility}
Here are additional details to the susceptibility analysis shown in the main text (Figure~\ref{fig:panel3_results}a). We draw 10,000 samples with payoff parameters $b, c \sim \text{Unif}(0.01,\, 0.99)$ and population frequency $x \sim \text{Unif}(0.05,\, 0.95)$. For each sample we fit both the constrained and free SBMs and record the per-edge log-likelihood difference and the spectral $\mathrm{SNR}$ of the planted block matrix. The $\mathrm{SNR}$ follows the Kesten–Stigum criterion \citep{decelle2011inference}: we construct the reduced matrix $C_{rs} = \sqrt{n_r\, n_s}\, B_{rs}$ and define $\textrm{SNR} = |\lambda_2| / \sqrt{\langle k \rangle}$, where $\lambda_2$ is the subdominant eigenvalue and $\langle k \rangle$ the mean degree. Community structure is information-theoretically detectable when $\textrm{SNR} > 1$.

\clearpage
\section{Parameter Recovery} \label{SI:parameter_recovery}
We test whether our \ours~can successfully recover the network structure implied by planted payoffs using synthetic data. Specifically, we vary $b$ and $c$ over a uniform $30\times30$ payoff grid, use the \FW~ to produce the implied edge-type distribution matrix $\mathbf{M}$, generate synthetic graphs with balanced communities ($x=0.5$), $N=500$ nodes, target mean degree $\langle k \rangle=10$, and fixed $\ell = 3.0$. We convert $\mathbf{M}$ to $\mathbf{B}$ using the density convention $E = \langle k \rangle \times N / 2$ from which we then sample graphs as independent Bernoulli draws per block. We then fit \ours~to the synthetic graph using EM restarts (10) to recover $\hat{\mathbf{M}}$. We measure structural recovery using the Frobenius norm

\begin{equation}
\left\| \mathbf{M}_{\mathrm{recovered}} - \mathbf{M}_{\mathrm{true}} \right\|_F
=
\sqrt{\sum_{i,j}
\left(
\mathbf{M}_{ij}^{\mathrm{rec}} - \mathbf{M}_{ij}^{\mathrm{true}}
\right)^2}
\label{eq:matrix_norm}
\end{equation}
of the better-aligned of the two label permutations, averaged over $B=100$ replications per cell.

Results are in Figure \ref{fig:recovery_heatmap}. Recovery error $\|\hat{\mathbf{M}} - \mathbf{M}\|_F$ is near zero throughout the detectable regime and spikes precisely at the Kesten–Stigum boundary, confirming that the constrained SBM recovers the planted partner-choice parameters wherever the community signal is information-theoretically identifiable.

\begin{figure}[h!]
    \centering
    \includegraphics[width=.6\linewidth]{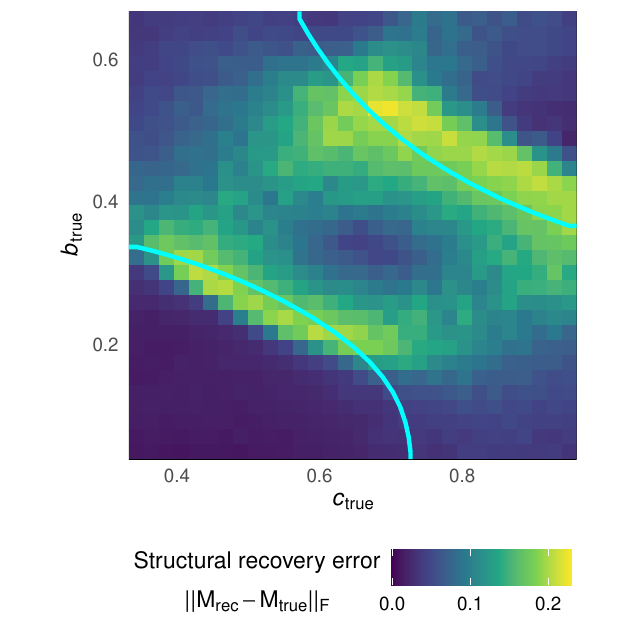}
    \caption{
        \textbf{Recovery error ($\|\hat{\mathbf{M}} - \mathbf{M}\|_F$) and $\mathrm{SNR}=1$ phase boundary.} 
        We restrict to the delight quadrant ($b<a,c>d$), the only region containing the random-mixing fixed point $(b,c)=(d,a)$. Crucially, parameter recovery error highlights the phase boundary, not the unidentifiable regime itself. Deep inside the unidentifiable phase---such as exactly at the random-mixing fixed point---the \ours trivially but accurately recovers the uniform random-mixing matrix, yielding low error despite the true community labels being information-theoretically unidentifiable. The metric is insensitive where $c \approx a$ because the X-row of $\mathbf{M}$ approaches its random-mixing value. Instead, recovery error spikes strictly along the $\mathrm{SNR} = 1$ critical transition boundary, where the latent signal is strong enough to trigger detection but heavily obscured by noise, demonstrating that \ours accurately recovers the underlying statistical physics of the phase transition.}
    \label{fig:recovery_heatmap}
\end{figure}

\input{SI_causal_emergence}
\label{causal_emergence}

\clearpage
\section{Datasets}
\input{table_datasets}

\clearpage
\section{Empirical Results}\label{SI:empirical_results}
The table below shows the main empirical results, using the EM estimator. 

\paragraph{Sensitivity to $\ell = 3.0$.} Results in the table above use $\ell$ fixed at $3.0$. This value represents systems that are neither fully fossilized ($\ell \rightarrow \infty$) nor random mixing ($\ell = 0$). We find the fitting procedure is not very sensitive to the exact value of $\ell$. To validate this we implemented a scan procedure that probes different values of $\ell$ using Brent's method with $\ell \in [0.5, 10.0]$ and selects the best-fitting model based on ELBO. We find virtually identical fits for different values of $\ell$. For example, the procedure identifies best-fitting $\ell = 3.99$ for the Zachary dataset with spread between highest and lowest value on the order of $10^{-5}$\%.

\begin{landscape}

\input{table_results_datasets}
\end{landscape}

\clearpage
\input{SI_gibbs}

\clearpage
\section{Forecasting}
Table \ref{tab:trade_forecasting} shows the trade network forecasting results.
\input{table_trade_forecasting}

\section{Model Selection through Topological Manifold Collapse}
As a visual diagnostic, Figure \ref{fig:panel3_results}b shows a threshold dividing the peripheral distances between collapsed $K=2$ models and distinct $K=3$ models. We derive this threshold empirically by plotting peripheral distances across many restarts and pick a point dividing the bimodal distribution (Figure \ref{fig:implicit_model_selection}).

\begin{figure}[h!]
    \centering
    \includegraphics[width=1\linewidth]{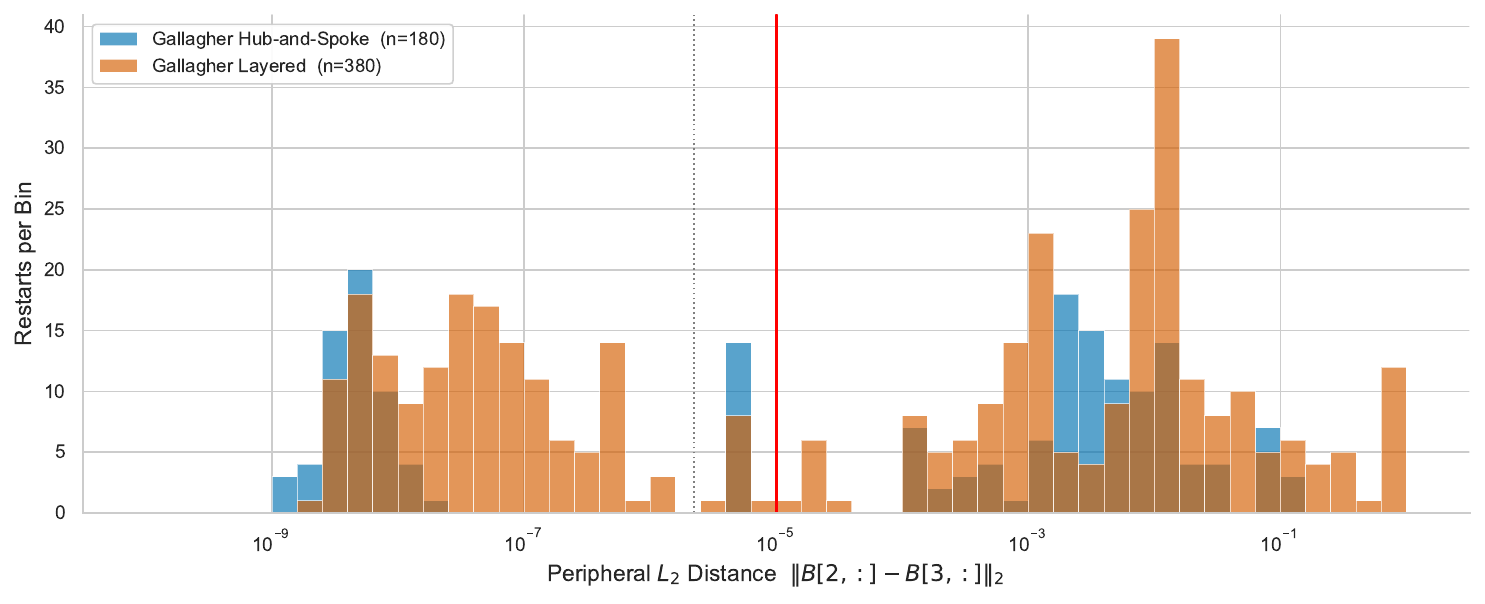}
    \caption{
    \textbf{Distribution of distances in peripheral blocks across EM restarts.}
    We fit 560 partner-choice EM restarts (180 Hub-and-Spoke, 380 Layered) on the Gallagher Communications networks, and plot the peripheral $L_2$ distances $|B[2,:]-B[3,:]|_2$ on a log axis. The distribution is clearly bimodal: a low-distance mode near the numerical floor ($10^{-9}$–$10^{-6}$), where the two peripheral $B$-rows collapse onto each other (degenerate $K=2$-effective basin), and a high-distance mode ($10^{-3}$–$10^{-1}$) where they remain distinct (genuine 3-block solution). The bimodality appears in both Gallagher labels, ruling out a Hub-vs-Layered artifact. The valley between the modes motivates the visual collapse threshold of
  $10^{-5}$ (red line).}
    \label{fig:implicit_model_selection}
\end{figure}

\clearpage
\section{Deplatforming Experiment}\label{si:deplatform}
We perform an intervention diagnostic experiment using the PolBlogs dataset from \citep{adamic2005political}. The network shows linking pattern of political blogs with two dense clusters. Clusters map on political leaning \citep{adamic2005political}. We use the directed version of the network (Table \ref{tab:empirical_results} shows the undirected network). \ours fits the data well, decisively outperforming the free SBM ($\Delta\Sigma_{\mathrm{BIC}} = -3.2$ nats, $\Delta\Sigma_{\mathrm{Lap}} = -25.2$ nats).

For the experiment, we sort nodes by degree and remove them one by one. After every step, we re-fit \ours and plot the inferred $c$ payoff (with the usual asymmetric diagonal pinned at $a=2/3$ and $d=1/3$). While the convention $c/b \le a/d$ pins the $X$ cluster to ``liberal'' as the more homophilous group on the full network, this pinning is not guaranteed once we start changing the network by removing nodes. Hence, we pin $X{=}$``liberal'' semantically using the ground-truth labels provided in the dataset to ensure the inferred $c$ payoff plotted in Figure \ref{fig:panel3_results}c speaks to the cross-community payoff of the same cluster across all node-removal steps. 

The experiment demonstrates that node removal does not easily alter the inferred payoff matrix. Removing nodes changes the block frequencies and therefore the strategic structure the network reveals, not the underlying incentives. The transition from a coordination game to a synergy game should thus be interpreted as a change in identifiability. Once the central within-camp connectors are removed, the surviving network can no longer be explained by within-camp-preferential linking and is best fit by a game in which cross-camp ties are rewarded. This is precisely why deplatforming is a limited lever to break up echo-chambers: because the incentives are untouched, the coordination structure would re-form under the same payoffs, and durable restructuring requires an intervention that moves the payoffs themselves across a boundary rather than merely removing the nodes that evidence the current game. This offers useful insights for interventions on deplatforming and echo chambers \citep{mccabe2024post,cinelli2021echo}. The targeted-versus-random contrast confirms the transition is a systematic property of which nodes carry the strategic structure, not finite-sample noise. Random removal, which perturbs composition without targeting the coordinators, never produces an ordered crossing. The case study therefore demonstrates the difference between altering a network's realized topology and altering the incentives that generate it---the distinction at the heart of our intervention analysis.

%% file: SI_analytical_results.tex
\subsection*{Derivation of the partner-choice equations from Game-Graph Dynamics}
\label{si:analytical}

The Game-Graph Dynamics are governed by the affinity function (Eq.~\ref{eq:affinity}) and the replicator update rule (Eq.~\ref{eq:graph-dynamics}; see Materials and Methods). Here we derive the closed-form partner-choice equations that result from iterating this dynamic for $\ell$ generations starting from random mixing.

\paragraph{Closed form.}
Affinity $\alpha_{ij} = U_{ij} + w$ is constant with respect to $q_{ij}$. The replicator (Eq.~\ref{eq:graph-dynamics}) therefore has constant fitnesses, and its closed-form solution after $\ell$ generations starting from random mixing $q_{ij}^{(0)} = x_j$ is
\begin{equation}
q_{ij}^{(\ell)} = \frac{x_j (U_{ij} + w)^{\ell}} {\displaystyle\sum_k x_k (U_{ik} + w)^{\ell}}.
\end{equation}
With $w > 0$ the solution involves $(U_{ij} + w)^\ell$. Since $w$ is a uniform additive shift, it does not affect the ordinal ranking of payoffs ($U_{ij} + w > U_{ik} + w \iff U_{ij} > U_{ik}$) and therefore leaves all phase boundaries and topology classifications unchanged. Without loss of generality, we set $w = 0$ for structural analysis. The joint mixing probability $M_{ij} = x_i q_{ij}^{(\ell)}$ then simplifies to
\begin{equation}
M_{ij} = \frac{x_i x_j U_{ij}^{\ell}}{\displaystyle\sum_k x_k U_{ik}^{\ell}}.
\label{eq:mixing_general}
\end{equation}
For the $2\times 2$ case, Equation~\ref{eq:mixing_general} gives Equations~\ref{eq:physics}. The parameter $\ell$ counts the number of generations of partner-choice learning starting from random mixing. At $\ell = 1$ the probabilities are payoff-proportional; as $\ell \to \infty$ the system concentrates all mass on the highest-payoff partner type.

\subsection*{Results for $K = 2$ games}

The propositions below establish which network topologies arise from \textit{two-type games} for any monotone selection dynamic.

A \textit{monotone selection dynamic} on link probabilities is any continuous dynamic, in discrete or continuous time, under which link weights to most-preferred partners never fall and grow whenever improvement is possible: if $\alpha_{ij} \geq \alpha_{ik}$ for all $k$, then $q_{ij}$ is non-decreasing along every trajectory, and increases strictly from any state where $q_{ij} > 0$ and $\alpha_{ij} > \bar{\alpha}_i$. The replicator Eq.~\ref{eq:graph-dynamics} is one such dynamic, as is any payoff-monotone imitative dynamic and the Smith dynamic \citep{hofbauer1998evolutionary,samuelson1992evolutionary}.


\vspace{12pt}
\begin{proposition}[Four-topology classification]
\label{si:result1}
For any $2 \times 2$ game, the asymptotically stable rest points of any monotone selection dynamic on link probabilities fall into exactly four topological classes: assortative (Pair Same), disassortative (Pair Different), X-absorptive (Pair X), and Y-absorptive (Pair Y).
\end{proposition}
\begin{proof}
At a rest point of any monotone selection dynamic, $\alpha_{ij} = \bar{\alpha}_i$ for all $j$ with $q^*_{ij} > 0$ and $\alpha_{ij} \leq \bar{\alpha}_i$ for all $j$ with $q^*_{ij} = 0$. A corner rest point $q^*_{ij^*} = 1$ is asymptotically stable whenever $\alpha_{ij^*} > \alpha_{ij}$ for all $j \neq j^*$.

In the $2\times 2$ case each type faces a single binary preference: X prefers X ($\alpha_{XX} > \alpha_{XY}$) or Y ($\alpha_{XX} < \alpha_{XY}$); Y prefers X ($\alpha_{YX} > \alpha_{YY}$) or Y ($\alpha_{YX} < \alpha_{YY}$). The Cartesian product of these two binary decisions gives exactly four topological classes: Pair Same (assortative), Pair Different (disassortative), Pair X (X-absorptive), and Pair Y (Y-absorptive). This enumeration depends only on the sign structure of $\alpha_{ij} - \alpha_{ik}$, not on the specific dynamic.

A mixed rest point would require $\alpha_{ij} = \bar{\alpha}_i$ for both partners simultaneously. By Proposition~\ref{si:result2}, this occurs only on the measure-zero locus $a = c$ or $b = d$. For generic payoffs, therefore, the only asymptotically stable rest points are corners, and they realize exactly the four classes above.
\end{proof}

\vspace{12pt}
\begin{proposition}[Phase-boundary independence]
\label{si:result2} 
For any $2 \times 2$ game, the boundaries separating the four topological classes are $a = c$ and $b = d$: they are determined solely by the ordinal ranking of payoffs within each row, independent of the population composition $x$.
\end{proposition}

\begin{proof} With $\alpha_{ij} = U_{ij} + w$, the affinities do not depend on the population composition. The boundary $\alpha_{XX} = \alpha_{XY}$ is therefore $(a + w) = (c + w)$, i.e.\ $a = c$; the additive shift $w$ cancels. An identical calculation for $\alpha_{YX} = \alpha_{YY}$ yields $b = d$. Since $\alpha_{ij}$ is independent of $x$, so are the boundaries.
\end{proof}

The topology of the equilibrium network is therefore determined solely by the ordinal ranking of payoffs within each row of the payoff matrix. In particular, the game-space/graph-space phase diagram (Figure~\ref{fig:panel1_conceptual}) is a property of payoff structure alone.

\paragraph{Order parameter.}
Each phase boundary is the sign change of a per-row preference ``magnetization'' $m_r=q_{rr}-q_{r\bar r}=2q_{rr}-1$, which jumps to $\operatorname{sgn}(U_{rr}-U_{r\bar r})$ as $\ell\to\infty$. The two diagnostics we use in our analyses are orthogonal projections of $\mathbf m=(m_X,m_Y)$: $P(\text{pair same})=\tfrac12+\tfrac12(x m_X+y m_Y)$ (assortativity axis, Figure~\ref{fig:panel2_physics}a) and affinity imbalance $M_{XY}-M_{YX}=\tfrac{x-y}{2}-\tfrac12(x m_X-y m_Y)$ (which-core axis, Figure~\ref{fig:affinity_imbalance}).

In balanced populations these reduce to the uniform magnetization $m_X+m_Y$ and the staggered magnetization $m_X-m_Y$, placing the four phases at the cardinal points of the (assortativity, which-core) plane---assortative and disassortative on the assortativity axis (which-core $=0$), $X$- and $Y$-absorptive on the which-core axis (assortativity $=0$)---so that neither order parameter alone resolves all four.



\subsection*{Extension to $K > 2$ games}

For $K>2$ the affinity ordering and payoff ordering may not coincide and the general equilibrium condition must be stated in terms of affinities Eq.~\ref{eq:affinity}: 
\begin{equation}
q_{ij} > 0 \quad \text{only if} \quad
\alpha_{ij} \geq \alpha_{ik} \quad \text{for all } k,
\label{eq:gg-equilibrium-general}
\end{equation}
This condition characterizes the asymptotically stable rest points of any monotone selection dynamic on link probabilities \citep{hofbauer1998evolutionary, samuelson1992evolutionary}.

\begin{proposition}[Dynamic stability of GGE for any $K$]
\label{si:prop-kgt2-hzero}
When the game is non-degenerate, a state is a GGE if and only if it is asymptotically stable under any monotone selection dynamic on link probabilities.
\end{proposition}

\begin{proof}
Since $\alpha_{ij} = U_{ij} + w$ is constant and non-degeneracy gives each type $i$ a unique maximizer $f(i) = \operatorname*{argmax}_j U_{ij}$, the GGE is the corner $q_{i,f(i)} = 1$ for all $i$.

\textit{(GGE $\Rightarrow$ asymptotically stable).}
Let $V = \sum_i (1 - q_{i,f(i)})$; then $V \geq 0$, with equality only at the GGE corner. Whenever $q_{i,f(i)} < 1$, type $i$ places weight on a strictly worse partner, so $\alpha_{i,f(i)} = \max_k \alpha_{ik} > \bar\alpha_i$. Any monotone dynamic then increases $q_{i,f(i)}$, so that summand strictly decreases while no other summand increases. Thus $V$ strictly decreases in a punctured neighborhood of the GGE, so $V$ is a strict Lyapunov function and the GGE is asymptotically stable.

\textit{(asymptotically stable $\Rightarrow$ GGE).}
Let $q^*$ be a rest point other than the GGE corner, so $q^*_{i,f(i)} < 1$ for some type $i$. If $q^*_{i,f(i)} > 0$, then $\alpha_{i,f(i)} > \bar\alpha_i$ (weight sits on a strictly worse partner), so any monotone dynamic strictly increases $q_{i,f(i)}$ and $q^*$ is not a rest point, a contradiction. If instead $q^*_{i,f(i)} = 0$, then $q^*$ is a rest point (the multiplicative update fixes a zero coordinate), but it is not asymptotically stable: any perturbation to $q_{i,f(i)} = \varepsilon > 0$ gives $\alpha_{i,f(i)} > \bar\alpha_i$, so monotonicity drives $q_{i,f(i)}$ away from $0$. In either case, no non-GGE rest point is asymptotically stable.
\end{proof}

\paragraph{GGE Stability with Equal Payoffs.}
\noindent Proposition~\ref{si:prop-kgt2-hzero} assumes all payoffs are non-equal (non-degeneracy). Here, we generalize the GGE-stability results to the degenerate cases where some payoffs are equal. For each type $i$, call the highest-utility partners $i$'s best set $S_i = \operatorname*{argmax}_j U_{ij}$. At a GGE (Eq.~\ref{eq:gg-equilibrium-general}) each type puts weight only on $S_i$. Denote the set of such equilibria by $\Gamma$. Because GGEs are such that $\alpha_{ij}=\bar\alpha$, each point on $\Gamma$ is a rest point of every monotone dynamic. When every best set is a single partner type, $\Gamma$ is the lone corner of Prop.~\ref{si:prop-kgt2-hzero}; when some best set contains a tie, $\Gamma$ is a continuum of equilibria. While individual GGEs are no longer asymptotically stable in these cases, the next two propositions show that Lyapunov stability is retained and that the GGE set is always reached.

\vspace{6pt}
\begin{proposition}[Lyapunov stability of the GGE set]
\label{si:gge_degenerate_stable}
Under any monotone selection dynamic on link probabilities, the GGE set $\Gamma$ is Lyapunov stable.
\end{proposition}
\begin{proof}
Take the Lyapunov function of Prop.~\ref{si:prop-kgt2-hzero} with the best set in place of the unique maximizer,
\[
V = \sum_i \Big( 1 - \sum_{j \in S_i} q_{ij} \Big),
\] 
so $V \geq 0$, with equality exactly on $\Gamma$. When every $S_i = \{f(i)\}$ this is the $V = \sum_i (1 - q_{i,f(i)})$ of Prop.~\ref{si:prop-kgt2-hzero}. Since $\alpha_{ij} = U_{ij} + w$ is constant, every best-set partner ($j \in S_i$) has $\alpha_{ij} = \max_k \alpha_{ik} \geq \bar\alpha_i$, so no monotone dynamic decreases $q_{ij}$, each best-set weight $\sum_{j \in S_i} q_{ij}$ is non-decreasing, and $V$ is non-increasing. Since the state space is compact, small $V$ forces proximity to $\Gamma$, so a trajectory on which $V$ never grows stays near $\Gamma$. Therefore, $\Gamma$ is Lyapunov stable.
\end{proof}
\noindent \textit{Remark on stable points.}
Each individual equilibrium point is itself Lyapunov stable, since starting near the point the outside weight is small, it converges to the best-set. Because each best-set link probability $q_{ij}$ ($j \in S_i$) is individually non-decreasing, it rises by at most that small amount, leaving the state near the equilibrium point.

\vspace{6pt}
\begin{proposition}[Global attraction to the GGE set]
\label{si:gge_degenerate_attract}
Under any monotone selection dynamic on link probabilities, every interior trajectory converges to $\Gamma$.
\end{proposition}
\begin{proof}
Fix a type $i$, and let $W_i = \sum_{j \notin S_i} q_{ij} = 1 - \sum_{j \in S_i} q_{ij}$ be its outside weight, the $i$-th summand of $V$. Because the trajectory is interior, each best-set link probability starts positive and is non-decreasing (see proof of Prop.~\ref{si:gge_degenerate_stable}). Hence $q_{ij}$ for $j \in S_i$ is bounded below by its initial value and stays positive along the trajectory, and by continuity, at every $\omega$-limit point. At any such state where $W_i > 0$, $\bar\alpha < \max_k \alpha_{ik}$, thus $q_{ij}$ strictly increases for $j \in S_i$, and $W_i$ strictly decreases. Since $W_i$ is non-increasing and bounded below, its values along the trajectory converge to a limit. And since the state space is compact, the trajectory has a non-empty forward-invariant $\omega$-limit set. As $W_i$ is a continuous function of the state, every $\omega$-limit point carries the same limiting value and $W_i$ equals the limit throughout the $\omega$-limit set. If the limit were positive, every point of that set would have $W_i>0$ and so strictly decrease $W_i$, which is impossible on a forward invariant set where $W_i$ is constant \citep{lasalle1976stability}. Hence, the limit is $0$ and $W_i \to 0$. As this holds for every type, $V \to 0$ and by the compactness argument of Prop.~\ref{si:gge_degenerate_stable}, the trajectory converges to $\Gamma$.
\end{proof}

\paragraph{Counting equilibrium topologies.}
The four-topology classification for $2\times 2$ games (Proposition~\ref{si:result1}) raises a natural question: how many distinct non-degenerate game-graph equilibrium topologies exist for general $K$-type games? A $K$-type game is \textit{non-degenerate} if every row of the payoff matrix has a unique maximum: for each $i \in [K]$, there is a unique $j^*(i)$ such that $U_{ij^*(i)} > U_{ij}$ for all $j \neq j^*(i)$. Non-degeneracy holds for Lebesgue-almost-all payoff matrices. When non-degeneracy holds, the GGE condition Eq.~\ref{eq:gg-equilibrium-general} requires each type $i$ to direct all interactions toward $j^*(i)$. The equilibrium topology is therefore completely encoded by the \textit{preference function}
\begin{equation}
    f : [K] \to [K], \qquad f(i) = \operatorname*{argmax}_{j \in [K]}\; U_{ij}.
    \label{eq:preference_function}
\end{equation}

\begin{proposition}[Labeled count]
\label{si:labeled_count}
There are exactly $K^K$ distinct non-degenerate game-graph equilibrium topologies over labeled $K$-type games, and every such topology is realizable.
\end{proposition}

\begin{proof}
In the non-degenerate case, the link probabilities $q_{ij}$ are fully determined by the $K$ values $f(1), \ldots, f(K)$: each type $i$ sends all links to $f(i)$, so $q_{i,f(i)} = 1$ and $q_{ij} = 0$ for $j \neq f(i)$. Since $f$ assigns each of the $K$ types an image in $[K]$, the total number of preference functions $f: [K] \to [K]$ is $K^K$.

For realizability, fix any $f: [K] \to [K]$ and set $U_{ij} = \mathbf{1}[j = f(i)]$. Then $f(i) = \operatorname*{argmax}_j U_{ij}$ for every $i$, so this payoff matrix is non-degenerate and realizes $f$. Since $f$ was arbitrary, every element of $[K]^{[K]}$ is realizable.
\end{proof}

\noindent The $K^K$ formula is the direct generalization of the $2^2 = 4$ four-topology result: each of the $K$ types independently selects a preferred partner from $K$ options. In empirical applications, type labels carry meaning (cooperator vs.\ defector, core vs.\ periphery), so the labeled count $K^K$ is the relevant quantity.

\paragraph{Strategic equivalence and manifold collapse.}
Proposition~\ref{si:labeled_count} counted the topologies generated by strategically distinct types; we now treat the degenerate case and show that a strategically equivalent (redundant) type is a null direction of the equilibrium map, so the partner-choice manifold carries no coordinate for it and the effective number of communities is fixed by the payoffs rather than chosen.

\vspace{6pt}
\begin{definition}[Strategic equivalence]
\label{def:strategic_equiv}
Two types $r, s \in [K]$ are \textit{strategically equivalent}, written $r \sim s$, if they are interchangeable both as \textit{partners} (every type earns the same payoff from $r$ as from $s$) and as \textit{choosers} (they rank their own partners the same way):
\begin{equation}
\underbrace{U_{jr} = U_{js}\ \text{for all } j}_{\text{equal payoff as a partner}}, \qquad
\underbrace{U_{rj} = \lambda\,U_{sj}\ \text{for all } j,\ \text{some } \lambda > 0}_{\text{proportional payoff as a chooser}}.
\label{eq:strategic_equiv}
\end{equation}
The two conditions are asymmetric because of how each is normalized. As a \textit{chooser}, type $r$'s outgoing links follow $q_{rj} \propto x_j U_{rj}^{\ell}/Z_r$; the own-normalizer $Z_r$ absorbs any overall scaling of $r$'s payoffs, so only their ratios matter and proportionality suffices (the model's magnitude gauge). As a \textit{partner}, $r$'s incoming links from any chooser $j$ follow $q_{jr}/x_r = U_{jr}^{\ell}/Z_j$, sharing $j$'s normalizer $Z_j$; with no scaling to absorb, $r$ and $s$ must yield exactly equal payoffs to each partner.
\end{definition}

\begin{lemma}[Identical mixing] \label{si:mixing_lemma}
Under the partner-choice map (Eq.~\ref{eq:physics}), $q_{ij}^{(\ell)} = x_j U_{ij}^{\ell}/Z_i$
with $Z_i = \sum_k x_k U_{ik}^{\ell}$, for any types $r,s$ and any $\ell \geq 1$: 

\begin{enumerate}
\item[\emph{(i)}] $r$ and $s$ have identical \emph{outgoing} distributions,
$q_{rj}^{(\ell)} = q_{sj}^{(\ell)}$ for all $j$, iff their chooser payoffs are
proportional, $U_{rj} = \lambda\,U_{sj}$ for all $j$, some $\lambda > 0$.
\item[\emph{(ii)}] $r$ and $s$ receive identical per-capita \emph{incoming} links,
$q_{jr}/x_r = q_{js}/x_s$ for all $j$, iff their partner payoffs are equal,
$U_{jr} = U_{js}$ for all $j$.
\end{enumerate}
\end{lemma}

\begin{proof}
\emph{(i)} $q_{rj}^{(\ell)} = q_{sj}^{(\ell)}$ for all $j$ iff $U_{rj}^{\,\ell} = (Z_r/Z_s)\,U_{sj}^{\ell}$ for all $j$, i.e.\ $U_{rj} = \lambda\,U_{sj}$ with $\lambda = (Z_r/Z_s)^{1/\ell} > 0$ (payoffs positive). \emph{(ii)} $q_{jr}/x_r = U_{jr}^{\,\ell}/Z_j$ and $q_{js}/x_s = U_{js}^{\,\ell}/Z_j$ share the denominator $Z_j$, so they are equal for all $j$ iff $U_{jr} = U_{js}$ for all $j$.
\end{proof}

Let $(\mathbf U,\mathbf x,\ell)$ be a game with $r\sim s$ (Definition~\ref{def:strategic_equiv}). Proportional chooser rows induce identical outgoing distributions (Lemma~\ref{si:mixing_lemma}(i)), so for the map $\mathbf M$ we may take $\lambda=1$ without loss of generality; then $U_{ri}=U_{si}$ for all $i$, the $\{r,s\}$ block of $\mathbf U$ is constant, and the reduced game $(\mathbf U^\pi,\dots)$ is well defined.

\vspace{6pt}
\begin{proposition}[Manifold collapse: no degenerate equilibrium] \label{si:manifold_collapse}

Let $(\mathbf U,\mathbf x,\ell)$ be a game with $r\sim s$ (Definition \ref{def:strategic_equiv}). Since proportional chooser rows induce identical outgoing distributions (Lemma~\ref{si:mixing_lemma}(i)), we may normalize $\lambda=1$ without loss of generality, so $U_{ir}=U_{is}$ and $U_{ri}=U_{si}$ for all $i$ and the $\{r,s\}$ block of $\mathbf U$ is constant.
Let $\pi$ merge $r,s$ into a single class $c$ with $x_c=x_r+x_s$, inducing the reduced game $(\mathbf U^\pi,\mathbf x^\pi,\ell)$ with $U^\pi_{ci}=U_{ri}$, $U^\pi_{ic}=U_{ir}$, $U^\pi_{cc}=U_{rr}$. Then the partner-choice map factors through $\pi$:
\[
\pi_\star\,\mathbf M(\mathbf U,\mathbf x,\ell)\;=\;\mathbf M(\mathbf U^\pi,\mathbf x^\pi,\ell),
\]
where $\pi_\star$ sums the $r$ and $s$ rows and columns of $\mathbf M$. Moreover the split ratio $x_r{:}x_s$ is a null direction of the map: reallocating class-$c$ mass between $r$ and $s$ at fixed $x_c$ leaves $\pi_\star\mathbf M$ unchanged.
\end{proposition}

\begin{proof}
The outgoing edge probability is $q_{ij}=x_jU_{ij}^{\ell}\big/\sum_k x_kU_{ik}^{\ell}$. For a chooser $i\notin\{r,s\}$, $U_{ir}=U_{is}=:U_{ic}$ gives $q_{ir}+q_{is}=(x_r+x_s)U_{ic}^{\ell}/Z_i$ with $Z_i=\sum_{k\notin\{r,s\}}x_kU_{ik}^{\ell}+x_cU_{ic}^{\ell}=Z_i^\pi$, and $q_{ij}=q^\pi_{ij}$ for $j\notin\{r,s\}$. 
The choosers $r,s$ have equal rows, hence equal denominators $Z_r = Z_s = Z_c^\pi$, making $q_{rj}=q_{sj}$ for $j\notin\{r,s\}$. For choices within the merged block, both $r$ and $s$ allocate the same combined probability: $q_{rr}+q_{rs} = q_{sr} + q_{ss} = x_c U_{cc}^\ell / Z_r = q_{cc}^\pi$. 
Summing $M_{ij}=x_iq_{ij}$ over the merged indices yields $\pi_\star\mathbf M=\mathbf M^\pi$. Every expression depends on $r,s$ only through $x_r+x_s$, so replacing $(x_r,x_s)$ by any $(x_r',x_s')$ with $x_r'+x_s'=x_c$ leaves $\pi_\star\mathbf M$ fixed; the family $\{x_r'+x_s'=x_c\}$ is a fiber of the map.
\end{proof}

\noindent\textit{Remark (dynamic-independence).}
The collapse is a property of the equilibrium, not of the replicator parameterization used above. By the partner condition $U_{ir}=U_{is}$, equivalent types present equal affinities to every chooser, so under any imitative (growth-rate-monotone) selection dynamic on link probabilities each chooser's split ratio $q_{ir}{:}q_{is}$ is a constant of motion: no such dynamic resolves the partition, and the rest point lies on the null fiber regardless of the dynamic. The closed form $q_{ij}=x_jU_{ij}^{\ell}/Z_i$ is the state of one such dynamic—$\ell$ multiplicative updates of the replicator from random mixing.

\noindent\textit{Remark (self-regularization and its sample image).}
Proposition~\ref{si:manifold_collapse} is parameter-free and holds before any model is fit: the $K$-type equilibrium with a redundant type pair realizes exactly the $(K{-}1)$-type topology. The population-level statement transfers to the likelihood through node exchangeability: the per-dyad edge probability from a type-$i$ node to any individual type-$r$ node is $q_{ir}/n_r \propto U_{ic}^{\ell}/Z_i$, which by the partner condition is identical for type-$s$ nodes and independent of the split $x_r{:}x_s$. Nodes of the two equivalent types are therefore exchangeable, and the observed-data likelihood is exactly constant along the null fiber: the split is unidentifiable. This flat direction is the exact source of the rank-deficient Fisher information observed in finite-sample fits---the empty hard block and singular curvature are the sample image of the null fiber. A fit with $K$ blocks can still represent the redundant pair as separate blocks, but only with an arbitrary split and duplicated block parameters at the same likelihood as the reduced $(K{-}1)$-block equilibrium; the extra block carries no information, so no description-length or BIC penalty is needed to \textit{justify} the smaller $K$. Which point of the fiber an optimizer reaches is then decided by initialization and finite-sample tie-breaking rather than by the population geometry; the empirical results (Results, ``Self-Regularization: Automatic Block Selection''; Figure~\ref{fig:panel3_results}b) show that EM restarts land in the collapsed, empty-block basin in practice. The exact non-identifiability is what makes this self-regularization possible, not a numerical accident.

\textit{Take-away.} A strategically redundant type never constitutes a sustained distinct community and mechanistic inference provides self-regularization. A redundant type is observationally null---a flat fiber of the equilibrium map that the inference collapses exactly. Redundant types collapse onto the strategic-equivalence quotient.

\vspace{6pt}
\noindent Throughout, $\bU>0$ entrywise and $x_i\ge x_{\min}>0$ for all $i$. We relax the non-degeneracy of Proposition~\ref{si:prop-kgt2-hzero} to its \textit{quotient} form: every row of the reduced game $(\bU^\pi,\mathbf x^\pi,\ell)$ on the classes of $\sim$ has a unique maximizer, written $[f(i)]$. Ties \textit{within} equivalence classes are permitted (they are forced by $\sim$), so this assumption, Proposition~\ref{si:manifold_collapse}, and the GGE-set results (Propositions~\ref{si:gge_degenerate_stable}--\ref{si:gge_degenerate_attract}) all apply to the same game.
\begin{proposition}[Block persistence]
\label{si:persistence}
Let $\mathcal{G} = (\bU, K, \mathbf{x}, \ell)$ satisfy quotient non-degeneracy and let types $r, s$ be \textit{not} strategically equivalent ($r \not\sim s$). Then:
\begin{enumerate}
\item[\emph{(i)}] \textbf{Distinct structure.} If $r$ and $s$ have non-proportional chooser payoffs, then $q_{rj}^{(\ell)} \neq q_{sj}^{(\ell)}$ for some $j$ at every finite $\ell \ge 1$; if their partner payoffs differ ($U_{jr}\neq U_{js}$ for some $j$), the per-capita incoming rates from $j$ differ. Both conditions are also necessary (Lemma~\ref{si:mixing_lemma} and the converse in the proof), and the index $j$ may lie in $\{r,s\}$.
\item[\emph{(ii)}] \textbf{Selection differential.} If $[f(r)] \neq [f(s)]$, types $r, s$ converge to GGE points supported on disjoint classes as $\ell \to \infty$, so their separation is bounded below for all large $\ell$; if they share a best class but differ in their chooser-payoff ratios, their outgoing distributions stay distinct at every finite $\ell$, though the separation decays geometrically as both concentrate on the shared class; and if $U_{jr} \neq U_{js}$, type $j$ has a strict affinity differential $\alpha_{jr} \neq \alpha_{js}$.
\item[\emph{(iii)}] \textbf{Likelihood separation.} Because $r \not\sim s$, merging them misspecifies at least one block-pair edge probability (by~(i), an outgoing $B_{rj} \neq B_{sj}$ or an incoming $B_{jr} \neq B_{js}$). The resulting per-pair KL divergence is strictly positive, and the total information loss scales as $O(N^2)$ (dense regime; $\Theta(N)$ when $B = \Theta(1/N)$) against an $O(\log N)$ BIC penalty for the $2K - 1$ extra block parameters, so the $K$-block model is preferred to the merge, with probability tending to one, above a game-dependent network size $N^\ast(\mathcal G)$.
\end{enumerate}
\end{proposition}
\begin{proof}
Work with general $K$. \emph{(i)} From the closed form $q^{(\ell)}_{ij}=x_jU_{ij}^{\ell}/Z_i$: if $q_{rj}^{(\ell)} = q_{sj}^{(\ell)}$ for all $j$, then $U_{rj}^{\ell}/Z_r = U_{sj}^{\ell}/Z_s$, so $U_{rj} = (Z_r/Z_s)^{1/\ell}\,U_{sj}$ for all $j$ ($\bU>0$), i.e.\ the chooser rows are proportional; the contrapositive gives the outgoing claim, and the converse direction is Lemma~\ref{si:mixing_lemma}(i). The per-capita incoming rate from chooser $j$ to a type-$r$ node is proportional to $U_{jr}^{\ell}/Z_j$ with denominator common to $r$ and $s$, so the rates differ iff $U_{jr}\neq U_{js}$. (At $\ell=0$ all outgoing distributions equal $\mathbf x$, hence the restriction $\ell\ge1$.) \emph{(ii)} If the best classes differ, the replicator drives the weight on $[f(r)]$, resp.\ $[f(s)]$, to $1$ (Proposition~\ref{si:prop-kgt2-hzero} applied to the quotient game, non-degenerate by assumption; equivalently, Proposition~\ref{si:gge_degenerate_attract} on the full game), so $q_r^{(\ell)}, q_s^{(\ell)}$ approach faces with disjoint supports, the finite-$\ell$ differential $\alpha_{r,f(r)} - \alpha_{r,j} = U_{r,f(r)} - U_{rj} > 0$ driving the separation. If the best classes coincide but the chooser-payoff ratios differ, the chooser payoffs are not proportional, so~(i) applies at every finite $\ell$; the separation, however, is $O(\rho^{\ell})$ with $\rho = \max\{U_{ij}/U_{i,f(i)} : i\in\{r,s\},\, j\notin[f(i)]\} < 1$. For unequal partner payoffs, $\alpha_{jr} - \alpha_{js} = U_{jr} - U_{js} \neq 0$ at $h = 0$, a direct selection pressure. \emph{(iii)} The merged model assigns block $m$ a single outgoing/incoming profile, fit at the mass-weighted average, e.g.\ $\tilde B_{mj} = (x_r B_{rj} + x_s B_{sj})/(x_r + x_s)$ and similarly for incoming and within-block entries; by~(i), $r \not\sim s$ forces $B_{rj} \neq B_{sj}$ or $B_{jr} \neq B_{js}$ for some $j$, and since $x_r, x_s \ge x_{\min} > 0$ the average differs from at least one true entry. The per-pair loss $D_{\mathrm{KL}}(B_{rj} \,\|\, \tilde B_{mj}) > 0$ (strict convexity), summed over the $\ge x_{\min}^2 N^2(1+o(1))$ affected dyads, exceeds the $O(\log N)$ penalty above a finite $N^*$; in the sparse regime each divergence is $\Theta(1/N)$ and the total $\Theta(N)$, still dominating. The dyad log-likelihood ratio is a sum of independent bounded terms and concentrates around this mean (Bernstein), giving preference with probability tending to one. 
Since the per-pair divergence $\delta$ depends on $\ell$ through $\mathbf B$, part~(ii) implies $N^\ast(\ell)\to\infty$ for shared-best-class pairs.
\end{proof}

\noindent \textit{Remark (scope).} Part~(iii) rejects the specific $\{r,s\}$-merge; ranking the $K$-block model above \textit{every} coarser node partition follows from the standard SBM argument, using that (i) rules out proportional rows and equal columns on the quotient.

\noindent \textit{Detectability caveat.} Likelihood separation is a population-level statement; recovering the partition from a single finite network additionally requires the spectral $\mathrm{SNR}$ to exceed the Kesten--Stigum threshold of $1$---binding precisely in the sparse regime, where the separation in~(iii) is $\Theta(N)$---which we translate into payoff coordinates in Section~\ref{si:detectability}.

\noindent \textit{Remark (minimal warranted partition).} Propositions~\ref{si:manifold_collapse}--\ref{si:persistence} together characterize the minimal warranted $K$-block partition of a strategic network: the quotient $[K]/{\sim}$ of the type space under strategic equivalence. Redundant blocks collapse onto the quotient, while non-equivalent blocks persist as long as the network is large enough for the signal to exceed the detectability threshold. The \FW thereby answers the two central questions of community detection: \textit{what is a community?} (a strategic equivalence class) and \textit{how many are there?} (the number of classes under $\sim$, which the inference can recover without enumerating $K$---the collapse is exactly penalty-free, and EM restarts land in the collapsed basin empirically; Results, ``Self-Regularization: Automatic Block Selection''). The third, statistical question---\textit{when can we detect them?}---is answered when the spectral $\mathrm{SNR}$ of the payoff-derived reduced matrix exceeds the Kesten--Stigum threshold, which imposes strict topological ceilings on finite data (Section~\ref{si:detectability}).


%% file: SI_estimator_EM.tex
\section{Variational EM} \label{si:EM}
We fit the constrained model using variational EM, treating the block assignments as latent variables. Let $q_{ir}$ denote the variational probability that node $i$ belongs to block $r$.

\paragraph{E-step.} Holding the model parameters fixed, we update the soft assignments by the standard mean-field rule \citep{daudin2008mixture}. For the baseline unweighted Bernoulli model, this is:
\[
q_{ir} \propto x_r \prod_{j \ne i} B_{r,x_j}^{A_{ij}}\,(1 - B_{r,x_j})^{1 - A_{ij}},
\]
where the product is taken with respect to the current expected assignments. Block frequencies are re-estimated as $x_r = N^{-1}\sum_i q_{ir}$. Empirical applications utilizing degree-correction replace this with the standard Poisson SBM counterpart.

\paragraph{M-step.} Given the expected sufficient statistics $\langle m_{rs} \rangle$ and $\langle \bar{m}_{rs} - m_{rs} \rangle$ under $q$, we maximize the expected complete-data log-likelihood with respect to the payoff parameters $\phi$. In the free SBM, this step yields the closed-form update $\hat{B}_{rs} = \langle m_{rs}\rangle / \langle\bar{m}_{rs}\rangle$. Under our \FW, the hard constraint $B_{rs} = f(M_{rs}(\phi))$---where $M_{rs}$ is the endogenous partner-choice flow and $f(\cdot)$ is the network-appropriate link function---eliminates this shortcut. We instead solve the bounded optimization problem:
\[
\hat{\phi} = \arg\min_{\phi} \; \bigl[-\mathcal{Q}(\phi)\bigr], \qquad \mathcal{Q}(\phi) = \kappa \sum_{r,s}\bigl[\langle m_{rs}\rangle \log B_{rs}(\phi) + \langle \bar{m}_{rs} - m_{rs}\rangle \log(1 - B_{rs}(\phi))\bigr],
\]
using L-BFGS-B with box constraints on each payoff entry, where $\kappa = \frac{1}{2}$ for undirected networks and $\kappa = 1$ for directed networks.

\paragraph{Initialization.} To avoid poor local optima, we run the full EM from multiple random restarts. The first restart is initialized by spectral clustering of $\textbf{A}$ and the second by $k$-means on the leading eigenvectors; subsequent restarts use uniformly random soft assignments. Within each restart, an M-step is performed before the first E-step, translating the structured initialization directly into payoff parameters and preventing the initial E-step from collapsing the partition. The solution with the highest ELBO across all restarts is retained.

\section{Degree-Corrected Variant} \label{SI:degree_correction}
The model naturally extends to the degree-corrected SBM \citep{karrer2011stochastic} by introducing node-level propensity parameters $\theta_i$ and replacing the Bernoulli likelihood with its Poisson counterpart. Under this specification, the expected number of directed edges from node $i$ to node $j$ is modeled as $\lambda_{ij} = \theta_i \theta_j B_{\theta_i \theta_j}$. The E-step and M-step are modified accordingly; the game-theoretic constraint on \bB is unchanged.

%% file: SI_model_comparison.tex
\clearpage
\section{Model Comparison} \label{SI:model_comparison}
We compare the constrained and free SBMs using multiple criteria, evaluated from converged fits using hard class assignments $\hat{\boldsymbol{x}}$ (obtained by $\arg\max_r q_{ir}$ for variational EM) and full posterior draws (via MCMC sampling).

\paragraph{Description Length (BIC).} We define
\[
\Sigma_{\text{BIC}} = -\log P(\mathbf{A} \mid \hat{\boldsymbol{x}}, \hat{\mathbf{B}}) + L(\hat{\boldsymbol{x}}) + \tfrac{d}{2}\log |E|,
\]
where $L(\hat{\boldsymbol{x}})$ is the combinatorial cost of encoding the partition \citep{peixoto2017nonparametric} and $d$ is the number of free parameters. Lower $\Sigma_{\text{BIC}}$ indicates a better model.

\paragraph{Description Length (Laplace).} Replacing the BIC penalty with a Laplace approximation to the marginal likelihood yields
\[
\Sigma_{\text{Lap}} = -\log P(\mathbf{A} \mid \hat{\boldsymbol{x}}, \hat{\mathbf{B}}) + L(\hat{\boldsymbol{x}}) + \tfrac{d}{2}\log\frac{|E|}{2\pi} + \tfrac{1}{2}\log\det \mathbf{I}(\hat{\phi}),
\]
where $\mathbf{I}(\hat{\phi})$ is the Fisher information matrix. For the free SBM, $\mathbf{I}$ is diagonal and analytic; for the constrained model, we compute it by central finite differences of the negative log-likelihood with respect to the payoff parameters.

\paragraph{Vuong Test.} To test whether the two models are statistically distinguishable at the level of individual observations, we apply the Vuong likelihood-ratio test \citep{vuong1989likelihood}. Under the null hypothesis that both models are equally close to the true data-generating process, the test statistic is asymptotically standard normal, providing a $p$-value for model equivalence.


\paragraph{Bayesian Out-of-Sample Prediction.} As a fully Bayesian, non-asymptotic check, we evaluate pointwise out-of-sample prediction accuracy using the Watanabe-Akaike Information Criterion (WAIC) and Pareto Smoothed Importance Sampling Leave-One-Out Cross-Validation (PSIS-LOO-CV) \citep{vehtari2017practical} calculated from the MCMC posterior traces.

%% file: SI_detectability.tex
\section{Detectability Ceilings of the \oursfull}
\label{si:detectability}

The main text outlines two topological constraints on community detectability in the \ours: a structural ceiling on the signal-to-noise ratio ($\mathrm{SNR}$) for dilemma games, and a fundamental detectability nullspace across all selection intensities $\ell$. Both follow from the limiting structure of the joint edge-type distribution $\mathbf{M}$ implied by the partner-choice equations. We work under the homogeneous, regular-graph assumption with soft community sizes $n_X = xN$, $n_Y = (1-x)N$, target mean degree $\langle k \rangle$, and the reduced mixing matrix $C_{kl} = \sqrt{n_k n_l}\,B_{kl}$.

The manifold collapse (Proposition~\ref{si:manifold_collapse}) has an exact spectral fingerprint. In the reduced mixing matrix $\mathbf{C}$ equivalence of $r$ and $s$ makes their $\mathbf{C}$-profiles proportional (Lemma~\ref{si:mixing_lemma}), so $\mathbf{C}$ is rank-deficient on the $r$--$s$ contrast and the eigenvalue encoding their distinction is exactly zero; non-equivalence lifts it away from zero. The number of detectable communities is thus bounded by the count of eigenvalues exceeding the Kesten--Stigum threshold.

\paragraph{Symmetrization and the Reduced Mixing Matrix.}
For payoffs $\bigl[\begin{smallmatrix}a & c\\b & d\end{smallmatrix}\bigr]$, the partner-choice equations yield a joint edge-type distribution $\mathbf{M}$. Because spectral community detection via the Kesten--Stigum criterion evaluates undirected graphs, we analyze the symmetrized joint distribution $S = \frac{1}{2}(\mathbf{M} + \mathbf{M}^T)$.

To leading order in large networks, the reduced mixing matrix $\mathbf{C}$, which governs the spectral gap, is defined by $C_{kl} = \langle k \rangle \frac{S_{kl}}{\sqrt{x_k x_l}}$. Community structure is detectable if the second eigenvalue satisfies $\mathrm{SNR} = |\lambda_2(\mathbf{C})| / \sqrt{\langle k \rangle} > 1$.

\paragraph{Topological Ceiling in Dilemma Games.}
In dilemma games, as $\ell \to \infty$, the directed distribution converges to a rank-one, column-concentrated matrix:
\begin{equation}
\mathbf{M} \;\xrightarrow[\ell\to\infty]{}\;
\begin{bmatrix} x & 0 \\ 1-x & 0 \end{bmatrix}.
\end{equation}
Symmetrizing this directed topology yields:
\begin{equation}
\mathbf{S} = \begin{bmatrix} x & \frac{1-x}{2} \\ \frac{1-x}{2} & 0 \end{bmatrix}.
\end{equation}
Substituting $\mathbf{S}$ into the definition for $\mathbf{C}$ provides the limiting reduced mixing matrix:
\begin{equation}
\mathbf{C} = \langle k \rangle \begin{bmatrix} 1 & \frac{1}{2}\sqrt{\frac{1-x}{x}} \\ \frac{1}{2}\sqrt{\frac{1-x}{x}} & 0 \end{bmatrix}.
\end{equation}
The eigenvalues of $\mathbf{C}$ are $\lambda = \frac{\langle k \rangle}{2} \left( 1 \pm \sqrt{1 + \frac{1-x}{x}} \right)$. 
For a balanced community ($x=0.5$), this yields $\lambda_2 = \langle k \rangle \frac{1-\sqrt{2}}{2}$. The resulting maximum signal is exactly $\mathrm{SNR}_{\max}^{\text{dilemma}} = \frac{\sqrt{2}-1}{2} \sqrt{\langle k \rangle} \approx 0.207 \sqrt{\langle k \rangle}$. Consequently, the structure is fundamentally undetectable if $\langle k \rangle < 4/(\sqrt{2}-1)^2 \approx 23.3$. 
The detectability limit $\langle k \rangle \approx 23.3$ is therefore a universal constant for this family of dynamics.

\paragraph{The Payoff-Ratio Nullspace ($ad = bc$) and Spurious Spectral Signals.}
A fundamental detectability nullspace arises when both types exhibit identical relative preferences for partners, collapsing any latent community structure into pure popularity bias. This occurs exactly along the locus of balanced within-row payoff ratios, $a/c = b/d$ (equivalently $ad = bc$).

Along this nullspace, both $X$ and $Y$ types distribute their outgoing links according to the exact same target probability distribution $q$. The directed joint edge-type distribution therefore collapses to a rank-one outer product of the source frequencies and the universal target preferences: $M_{ij} = x_i q_j$. Crucially, unless payoffs are completely flat ($a=c$ and $b=d$), the target preference differs from the population baseline ($q_j \neq x_j$), meaning $x_i q_j \neq x_j q_i$ and the matrix \bM\ is fundamentally asymmetric.

When evaluating undirected graphs, spectral community detection algorithms evaluate the symmetrized matrix $\mathbf{S} = \frac{1}{2}(\mathbf{M} + \mathbf{M}^T)$. Because \bM\ is an asymmetric rank-one outer product, its symmetrization $\mathbf{S}$ is mathematically forced to be rank-two with a strictly negative determinant. Consequently, substituting this $\mathbf{S}$ into the definition for $\mathbf{C}$ yields a reduced mixing matrix with a spurious, strictly negative second eigenvalue ($\lambda_2 < 0$). Uncorrected spectral methods will therefore falsely detect a structural signal, confounding the directed popularity imbalance (the entire population universally preferring one partner type) with true network disassortativity.

While the failure of uncorrected spectral methods in the presence of degree heterogeneity is a known structural artifact \citep{karrer2011stochastic}, our framework reveals its exact game-theoretic origin. We show that this spectral illusion is deterministically triggered along the payoff-ratio locus $ad = bc$, demonstrating how balanced payoff ratios alone spontaneously generate the macroscopic degree imbalances that confound standard community detection.

\paragraph{Degree Correction and Gauge Invariance.}
The analytical ceiling derived for dilemma games and the spurious disassortative signals along the $ad = bc$ nullspace represent the worst-case detectability under the most parsimonious (Bernoulli, non-degree-corrected) model class. At the nullspace, the true structural community signal is fundamentally zero---the mutual information between source and target types vanishes exactly. The degree-corrected Poisson SBM successfully navigates this. The degree parameters ($\theta_i$) absorb the prominent degree-imbalance and popularity components inherent to these rank-one topologies, correctly returning $\mathrm{SNR} = 0$ along the nullspace and extending the identifiable regime substantially beyond the uncorrected spectral limits elsewhere.

%% file: SI_game_taxonomy_table.tex
\begin{table}[h!]
\renewcommand{\arraystretch}{1.4}

\begin{tabular}{>{\raggedright\arraybackslash}m{2.6cm} C{3cm} C{3cm} C{3cm} C{3cm}}
\toprule
& \textbf{Assortative} & \textbf{Disassortative} & \textbf{X-Absorptive} & \textbf{Y-Absorptive} \\
& {\footnotesize $a>c,\; d>b$} & {\footnotesize $a<c,\; d<b$} & {\footnotesize $a>c,\; d<b$} & {\footnotesize $a<c,\; d>b$} \\
\midrule
\textbf{Coordination} &
  \textbf{Coordination} &
  \cellcolor{lightgray} --- &
  \textbf{Stag Hunt} &
  \cellcolor{lightgray} --- \\[-4pt]
{\footnotesize $a>b,\; c<d$} &
  {\footnotesize Standard / Alt} &
  \cellcolor{lightgray} &
  {\footnotesize Team / Solo} &
  \cellcolor{lightgray} \\[6pt]

\textbf{Conflict} &
  \cellcolor{lightgray} --- &
  \textbf{Anti-Coord} &
  \textbf{Hawk-Dove} &
  \cellcolor{lightgray} --- \\[-4pt]
{\footnotesize $a<b,\; c>d$} &
  \cellcolor{lightgray} &
  {\footnotesize Provider / Consumer} &
  {\footnotesize Dove / Hawk} &
  \cellcolor{lightgray} \\[6pt]

\textbf{Dilemma} &
  \cellcolor{lightgray} --- &
  \cellcolor{lightgray} --- &
  \textbf{Dilemma} &
  \cellcolor{lightgray} --- \\[-4pt]
{\footnotesize $a<b,\; c<d$} &
  \cellcolor{lightgray} &
  \cellcolor{lightgray} &
  {\footnotesize Producer / Exploiter} &
  \cellcolor{lightgray} \\[6pt]

\textbf{Delight} &
  \textbf{Synergy} &
  \textbf{Mutualism} &
  \textbf{Harmony} &
  \textbf{Dissonance} \\[-4pt]
{\footnotesize $a>b,\; c>d$} &
  {\footnotesize Mainstream / Niche} &
  {\footnotesize Host / Symbiont} &
  {\footnotesize Attractor / Satellite} &
  {\footnotesize Invasive / Native} \\

\bottomrule\\[-10pt]
\end{tabular}
\caption{
    \textbf{Mapping between the four $K=2$ network topologies (rows) and game-theory Nash equilibria (columns) and their semantically interpretable role labels.}
    }
\label{SI:roles}
\end{table}

%% file: SI_causal_emergence.tex
\clearpage
\section{Affinity Imbalance Signatures}
The four-topology classification (Result~1; Figure \ref{fig:panel1_conceptual}) establishes analytically that distinct games leave distinct type-level signatures. We visualize this purely as a diagnostic, to make two points concretely: (a) the same Nash equilibrium can generate networks spanning all four topologies, each with a distinct signature; and (b) one topology can be generated by four different Nash games, indistinguishable under a standard block model yet immediately separable in their recovered payoffs under the partner-choice model.

We plot the \textbf{affinity imbalance} $M_{XY} - M_{YX}$---the difference between the share of edges directed $X \to Y$ and the share directed $Y  \to X$---as a function of rationality $\ell$. It is dimensionless (a difference of edge proportions), bounded in $[-0.5,+0.5]$ at balanced populations ($x=0.5$), with the sign identifying which strategy type is the absorbing core (negative $\Rightarrow X$, positive $\Rightarrow Y$). We show networks from the grid in Figure \ref{fig:panel1_conceptual} with payoffs $b, c \in \{0.17, 0.5, 0.83\}$.

\textit{Same Nash, four topologies} (Figure \ref{fig:affinity_imbalance}a). 
Within the X-dominant (Delight) Nash region, the four topology classes trace distinct affinity imbalance curves. The Dissonance game shows positive affinity imbalance whereas the Harmony game shows negative. The Synergy and Mutualism games show near-reciprocal attraction but via oppositely leaning trajectories at finite rationality. All four occupy distinct regions of the axis.

\textit{Same topology, four Nashes} (Figure \ref{fig:affinity_imbalance}b). Four Nash games (Dilemma, Stag Hunt, Hawk-Dove, Harmony) chosen to land in the X-absorptive region trace four distinct curves on the negative side of the axis, ordered by absorption strength from Dilemma (most negative) to Harmony (least). The curves converge as $\ell \rightarrow \infty$ yet their trajectories at finite rationality have distinct fingerprints.

\begin{figure}[h!]
  \begin{overpic}[width=0.48\textwidth]{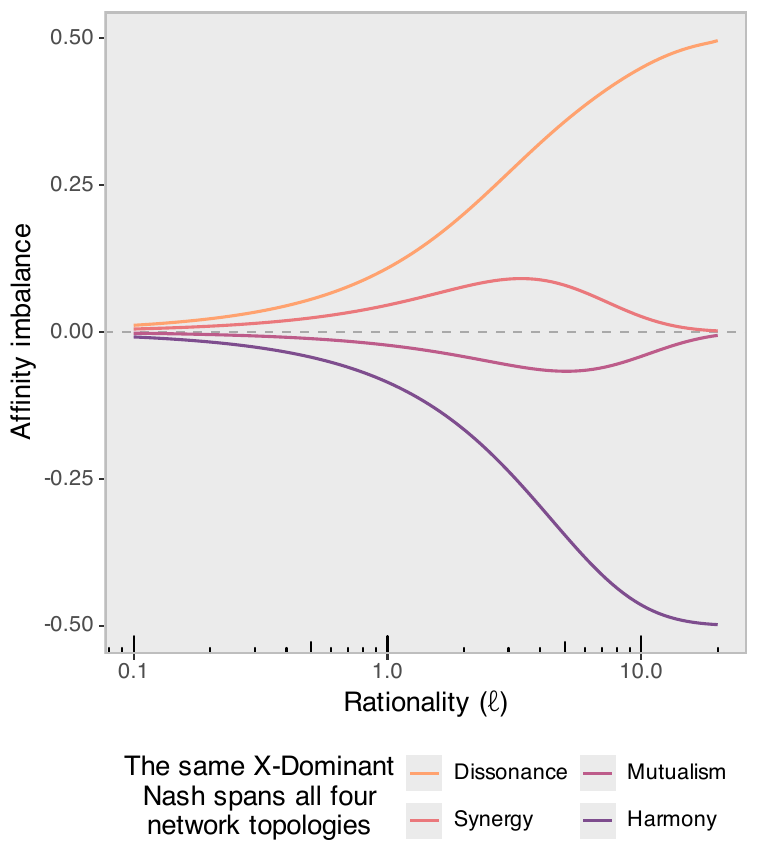}
    \put(2,96){\textsf{\textbf{a}}}
  \end{overpic}%
  \hfill
  \begin{overpic}[width=0.48\textwidth]{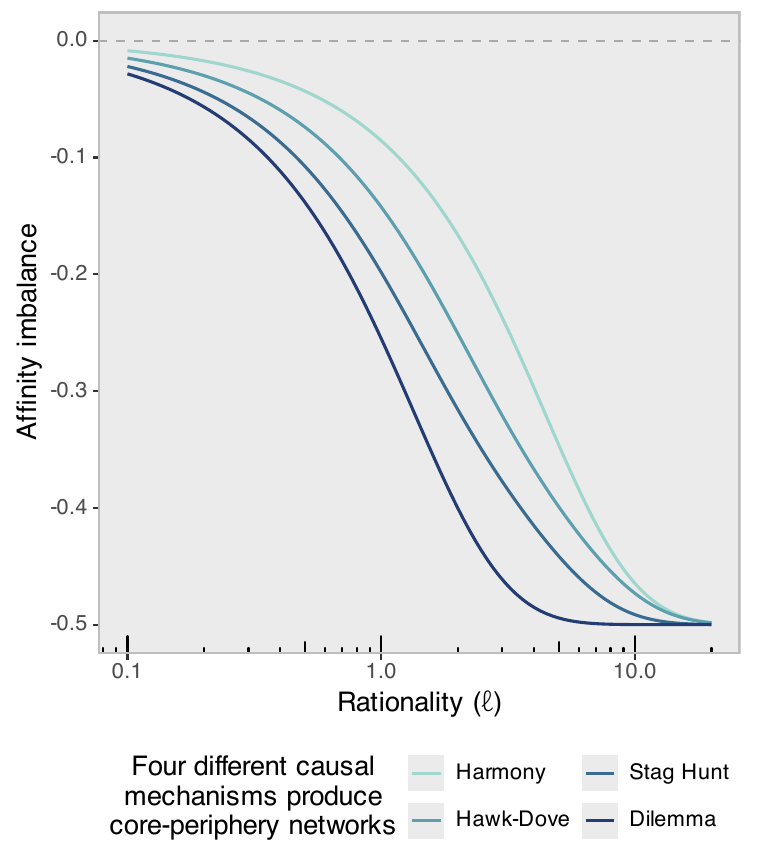}
    \put(2,96){\textsf{\textbf{b}}}
  \end{overpic}
    \caption{
        \textbf{The causal generative mechanism of \ours differentiates different emergent topologies.}
        \textbf{a} Four different instances of the same X-dominant Nash equilibirum, but with payoffs resulting in four different game-graph equilibria.
        \textbf{b} All four networks are core-periphery networks, but with payoffs from different games (different Nash).
    }
    \label{fig:affinity_imbalance}
\end{figure}

%% file: table_datasets.tex
\begin{table}[h!t]
\centering
\small
\setlength{\tabcolsep}{4pt}
\renewcommand{\arraystretch}{1.15}
\begin{tabular}{l l p{4.5cm} p{5.8cm}}
\toprule
\textbf{Dataset} & \textbf{Size} & \textbf{Source} & \textbf{Notes} \\
\midrule

Zachary karate club &
$N{=}34$, $E{=}78$ &
\texttt{nx.karate\_club\_graph()}; \citep{zachary1977information} & \\

Florentine families &
$N{=}15$, $E{=}20$ &
\texttt{nx.florentine\_families\_graph()}; \cite{breiger1986cumulated} & \\

Lusseau dolphins &
$N{=}62$, $E{=}159$ &
KONECT edge list (\texttt{dolphins}), \url{http://konect.cc/networks/dolphins/}; \citep{lusseau2003bottlenose} & \\

Les Mis\'erables &
$N{=}77$, $E{=}254$ &
\texttt{nx.les\_miserables\_graph()}; Knuth, \emph{The Stanford GraphBase} (1993) &  \\

PolBlogs &
$N{=}1{,}222$, $E{=}16{,}714$ & 
\citep{adamic2005political} &
Symmetrized and restricted to the largest connected component (1{,}222 of the 1{,}490 original nodes). \\

\textit{C.~elegans} PPI &
$N{=}2{,}528$, $E{=}3{,}706$ & \citep{wang2023assessment} &  \\

\textit{S.~cerevisiae} PPI &
$N{=}2{,}018$, $E{=}2{,}705$ & \citep{wang2023assessment} & \\

\textit{A.~thaliana} PPI &
$N{=}2{,}774$, $E{=}6{,}045$ & \citep{wang2023assessment} & \\

\textit{H.~sapiens} PPI &
$N{=}8{,}272$, $E{=}52{,}068$ & \citep{wang2023assessment} & \\

Trade &
\begin{tabular}[t]{@{}l@{}}
median:\\ $N{\in}[207, 228]$, \\
\hspace{1.2em} $E{\in}[5{,}194,\,9{,}475]$ \\
density:\\ $N{\in}[98, 144]$, \\
\hspace{1.2em} $E{\in}[977,\,1{,}862]$
\end{tabular} &
UN COMTRADE, 1994--2024 (31 years) &
Bilateral, country-level goods-trade flows, mirror-averaged across exporter and importer self-reports. For each of the 31 years, we construct a directed network by retaining each country's top-20 outgoing trade partners by volume. \\

\bottomrule\\[-7pt]
\end{tabular}
    \caption{
        \textbf{Datasets used in the empirical analyses.} 
        Sizes report node count $N$ and undirected edge count $E$ after the manipulation listed in the Notes column. Networks are undirected except the bilateral trade networks. For the trade panel, the ranges span all 31 years (1994--2024) under each thresholding strategy; the smallest year is 1994 in both, with peak density around 2013--2021.}
\label{tab:datasets}
\end{table}

%% file: table_results_datasets.tex
\begin{table}[htbp]
\centering
\footnotesize
\setlength{\tabcolsep}{4pt}
\begin{tabular}{@{}lll l l D{.}{.}{9.0} RRll@{}}

\toprule
Dataset & Model & Freqs & Payoffs $(b,\,c)$ & Game & ELBO & {$\Delta\Sigma_{\mathrm{BIC}}$} & {$\Delta\Sigma_{\mathrm{Lap}}$} & Vuong & Verdict \\
\midrule

\multirow{2}{*}{Florentine}
  & \ours & 0.454, 0.546 & $b{=}0.35,\; c{=}0.70$ & Mutualism       & $-49.44$        & $-1.57$ & --      & $0.870$        & Collapse to single group; near random mix    \\
  & Free SBM        & 0.510, 0.490 & --                      & --            & $-49.39$        &         &         &                   & random mix                \\
\midrule

\multirow{2}{*}{Zachary}
  & \ours & 0.484, 0.516 & $b{=}0.19,\; c{=}0.34$ & Synergy       & $-190.65$       & $-2.62$ & $-7.05$ & $0.610$        & Game-theoretic (strong)   \\
  & Free SBM        & 0.518, 0.482 & --                      & --            & $-190.48$       &         &         &                   &   \\
\midrule

\multirow{2}{*}{Dolphins}
  & \ours & 0.668, 0.332 & $b{=}0.14,\; c{=}0.01$ & Coordination  & $-466.54$       & $-2.59$ & $-11.23$& $0.768$        & Game-theoretic (decisive) \\
  & Free SBM        & 0.667, 0.333 & --                      & --            & $-466.30$       &         &         &                   &              \\
\midrule

\multirow{2}{*}{Les Miserables}
  & \ours & 0.137, 0.863 & $b{=}0.20,\; c{=}0.17$ & Coordination       & $-657.33$       & $-3.13$ & $-8.93$ & $0.156$        & Game-theoretic (strong)   \\
  & Free SBM        & 0.861, 0.139 & --                      & --            & $-657.08$       &         &         &                   &                           \\
\midrule

\multirow{2}{*}{Polblogs}
  & \ours & 0.458, 0.542 & $b{=}0.16,\; c{=}0.24$ & Coordination  & $-50{,}700.71$  & $-5.04$ & $-18.33$& $0.744$        & Game-theoretic (decisive) \\
  & Free SBM        & 0.458, 0.542 & --                      & --            & $-50{,}700.67$  &         &         &                   &                           \\
\midrule

\multirow{2}{*}{C.elegans}
  & \ours & 0.504, 0.496 & $b{=}0.33,\; c{=}0.67$ & Mutualism       & $-23{,}617.60$  & $-4.11$ & --      & $0.990$        & Undifferentiated payoffs (random mix) \\
  & Free SBM        & 0.503, 0.497 & --                      & --            & $-23{,}617.60$  &         &         &                   &     \\
\midrule

\multirow{2}{*}{Yeast}
  & \ours & 0.503, 0.497 & $b{=}0.33,\; c{=}0.67$ & Mutualism       & $-17{,}475.95$  & $-3.95$ & --      & $0.990$        & Undifferentiated payoffs (random mix) \\
  & Free SBM        & 0.500, 0.500 & --                      & --            & $-17{,}475.95$  &         &         &                   &       \\
\midrule

\multirow{2}{*}{A.thaliana}
  & \ours & 0.672, 0.328 & $b{=}0.12,\; c{=}0.19$ & Coordination  & $-33{,}641.18$  & $-4.72$ & $-25.36$& $0.816$        & Game-theoretic (decisive) \\
  & Free SBM        & 0.328, 0.672 & --                      & --            & $-33{,}641.10$  &         &         &                   &              \\
\midrule

\multirow{2}{*}{H.sapiens}
  & \ours & 0.266, 0.734 & $b{=}0.15,\; c{=}0.15$ & Coordination  & $-269{,}708.56$ & $-5.55$ & $-27.52$& $\mathbf{0.024}$ (c)  & Game-theoretic (decisive) \\
  & Free SBM        & 0.266, 0.734 & --                      & --            & $-269{,}708.55$ &         &         &                   &                           \\


\bottomrule\\[2pt]
\end{tabular}
    \caption{
        \textbf{Model comparison across network datasets.} 
        Estimates from EM, 1,000 restarts, $\ell=3.0$. Vuong test indicates whether the significant $p$-value is in favor of the constrained (c) or free (f) model.
$\Delta\Sigma_{\mathrm{Lap}}$ is undefined when the Fisher information at the MLE is not positive-definite (e.g., an empty block makes the free SBM's $M_{rs} = 0$, or the constrained model recovers undifferentiated payoffs ($c=a$ and $b=d$).}
\label{tab:empirical_results}
\end{table}

%% file: SI_gibbs.tex
\section{Fully Bayesian Inference and Model Validation} \label{si:gibbs}
To rigorously verify that our EM results are not artifacts of variational approximations or boundary singularities, we implement a fully Bayesian MCMC estimator. Because the non-linear mapping from game-theoretic payoffs to the edge probability matrix lacks a conjugate prior, we embedded a Metropolis-Hastings step within a standard Gibbs sampling routine. As with the EM estimator, we ensure identifiability by gauge-fixing the diagonal payoffs and address scaling degeneracy inherent to softmax choice models by fixing $\ell = 3.0$. We evaluate the posterior volume using the Watanabe-Akaike Information Criterion (WAIC) and Pareto Smoothed Importance Sampling Leave-One-Out Cross-Validation (PSIS-LOO-CV). We evaluate the Bayesian estimator on three synthetic datasets as well as the international trade network (from 2010). All synthetic datasets are directed and use $N=500$, $x=0.5$. We fit the Bayesian model with a fixed $\ell = 3.0$ and two chains and 1500 samples post-burn-in per chain (burn-in 750, thin 2) on the synthetic datasets and four chains and 5,000 on the sparser empirical trade network. 


Concretely, the sampler cycles through four conditional updates per iteration: (1) node labels $x_i$ are drawn categorically from their full conditional given the current edge-probability matrix and mixing weights; (2) mixing proportions $\boldsymbol{x}$ are drawn from a conjugate Dirichlet posterior; (3) in degree-corrected mode, the node propensities $\theta_i$ are drawn from conjugate Gamma posteriors with per-community renormalization to preserve identifiability with \bB; and (4) the off-diagonal payoffs are updated by a random-walk Metropolis–Hastings step. The game-theoretic constraint is enforced in this final step: each MH proposal perturbs the $K^2 - K$ free payoff entries and then deterministically maps the proposed payoffs through the softmax partner-choice equations and the appropriate link function to obtain the implied edge-probability matrix \bB, so \bB is never sampled freely but is always the image of the payoff surface. It means the sampler is exploring nonlinear, lower-dimensional manifold embedded in the full SBM parameter space, while the MH step is doing the heavy lifting to navigate that surface. The MH proposal variance is tuned during burn-in via Robbins–Monro adaptation targeting a 30\% acceptance rate (this tuning is very successful and observed acceptance rates are typically between 29-35\%), and five MH sub-iterations are executed per Gibbs sweep to improve mixing in the payoff subspace. A $\mathrm{Beta}(2,2)$ prior on each free payoff gently regularizes the posterior away from the $0/1$ boundary, which stabilizes inference on sparse or small graphs. Each chain is warm-started from a short EM fit at the same ($K, \ell$, and gauge fix) configuration, and cross-chain label switching is resolved post hoc by aligning all chains to a reference labeling via brute-force permutation search over the $K!$ relabelings. 

\paragraph{Positive Control: Constrained Ground Truth.} 
We generate a synthetic dataset with planted coordination payoffs ($b=0.1, c=0.3$) and fit our Bayesian Gibbs model. R-hats for the payoff estimates are 1.000 and ESS $[b=1809, c=2611]$. The model successfully recovers the coordination game structure ($a > c$ and $d > b$) and population frequencies almost perfectly $X = 0.4998$ [95\% CI: $0.498 - 0.502$], though individual payoffs shift along a sloppy direction in payoff space that leaves the implied \bB matrix nearly unchanged ($b = 0.217$ [95\% CI: $0.196 - 0.238$] and $c = 0.142$ [95\% CI: $0.132 - 0.153$]). WAIC prefers the constrained model by 4.32 on WAIC scale ($0.897$ posterior probability mass on the preferred model). PSIS-LOO-CV prefers the constrained model by 4.77 nats (LOOIC); with $0.916$ posterior probability mass (constrained ELPD$_\text{loo}$ = -12718.95 (SE: 206.37); free model ELPD$_\text{loo}$ = -12721.34 (SE: 206.30). About ~0.5\% of leave-one-out observations have Pareto $k > 0.70$ (slightly worse in the free model). The nat differences are consistent across both criteria. Even if the exact ELPD is slightly off due to PSIS instability, the ranking and magnitude agree with WAIC.

\paragraph{Negative Control: Free Ground Truth.} 
We generate a synthetic dataset that sits far off the lower-dimensional manifold of the game-theoretic constrained space, violently violating the modeling assumptions of \ours thus giving a clear advantage to the flexibility of the free SBM (implied \bM row sums $(0.25, 0.75)$ violate the constrained partner-choice equations' structural identity; block 1 emits $~3\times$ more out-edges than equal block sizes would predict). R-hats are 1.000 and ESS [1318, 2269]; however, PSIS reliability is poor as would be expected with 14\% of observations showing $k > 0.7$ given the constrained model is genuinely misspecified for this data. The free model is preferred with 1011.69 on the WAIC scale criterion and 992.20 nats using PSIS-LOO-CV---a massive advantage in favor of the free model. 

\paragraph{Neutral Control: ER Graph.} 
We fit both models to an ER random graph. R-hats are 1.002 and ESS [2238, 2198] and zero observations with Pareto $k > 0.70$. The free model has zero bad k values, the constrained has some ($\approx 2\%$) which is not surprising given node assignments are meaningless in this network). Both criteria agree on a small prediction advantage of the constrained model (WAIC = 1.83; and PSIS-LOO-CV = 2.21 nats). Furthermore, \ours collapsed all nodes into a single cluster with $X = 0.0007$ [95\% CI: $0.0001 - 0.0020$] and $Y = 0.9993$ [95\% CI: $0.9980 - 0.9999$], while the free SBM did not (otherwise LOO-CV would be identical). Furthermore, the predicted payoffs revert to their priors with very large confidence intervals ($c = 0.51$ [95\% CI: $0.06 - 0.99$], $b = 0.44$ [95\% CI: $0.08 - 0.91$]). Because \ours separates structural availability from behavioral preference, attempting to fit a two-block model to pure noise forces the preference parameters into an infinitesimally small volume of the parameter space (perfect indifference). The MCMC naturally rejects this, collapsing the network to a single block. This confirms that grounding the SBM in a strict game-theoretic choice dynamic provides a structurally built-in defense against overfitting that the unconstrained SBM lacks.

\paragraph{Expressiveness: Cyclic Dominance in $K=3$.}
To show the constrained model is not confined to assortative structure, we plant an intransitive rock--paper--scissors game on $K=3$ blocks: a directed cycle $X\!\to\!Y\!\to\!Z\!\to\!X$ with cycle payoffs $0.9$ and back-edge payoffs $0.05$. This structure carries no within-block density signal and is qualitatively different from what games with only two actions could represent. The Gibbs model recovers the cycle cleanly (R-hats $1.000$--$1.001$; ESS $875$--$2090$): the cycle payoffs $M_{XY}=0.923$ [95\% CI: $0.879 - 0.967$], $M_{YZ}=0.907$ [95\% CI: $0.855 - 0.962$], and $M_{ZX}=0.858$ [95\% CI: $0.782 - 0.948$] dominate their back-edges ($\leq 0.10$), and block frequencies recover the equal-sized blocks nearly perfectly ($x_x=0.333$, $x_y=0.335$, $x_z=0.332$). Both model fit criteria again prefer the constrained model---WAIC by $7.79$ ($0.980$ posterior probability mass) and PSIS-LOO-CV by $7.39$ nats (LOOIC; $0.976$ mass; constrained ELPD$_\text{loo}=-12051.52$ (SE: $195.72$), free $-12055.21$ (SE: $195.60$)), with $\sim$$0.6\%$ of leave-one-out observations at Pareto $k > 0.70$. 
This confirms the framework recovers intransitive, non-assortative organization and wins on parsimony.

\paragraph{International Trade Network (2010).}
On the international trade network for 2010 ($N=231$), the Bayesian model finds nearly identical $b, c$ payoffs and block frequencies as the EM model (Table \ref{tab:results_gibbs}). That is, the Bayesian integration resolved all boundary singularities encountered during optimization suggesting the game-theoretic constraint successfully regularize the parameter space yielding robust and theoretically sound EM estimates. 

This also confirms our core substantive insights about the trade network structure: the game-theoretic model infers a directed core-periphery with an asymmetric extraction game dynamic (i.e., an unequal exchange game). The core trades with the periphery, extracting a premium over trading internally ($b > a$); the core shows highly profitable within-core trade but lacks the exploitative premium of trading with the periphery (e.g., cheap labor and raw materials). Countries in the periphery receive only a low payoff from trading with other peripheral countries (i.e., trade between developing nations, constrained by lower capital and purchasing power). Periphery trades with the core yields the lowest payoff resulting from deteriorating terms of trade. ``Core'' is a strictly dominant strategy ($a>c$ and $b>d$). While the core strategy is appealing in economic terms, it is driven by industrialization, technological supremacy, and capital accumulation and peripheral countries cannot simply choose to adopt ``core'' as their strategy as it requires overcoming significant barriers to entry.
The constrained and free models achieve nearly identical predictive fit: $\Delta \text{WAIC} = +1.18$ nats in favor of the free model, while $\Delta \text{LOOIC} = -3.99$ nats in favor of the constrained model. 

R-hats of 1.007 indicate that chains have converged to the same stationary distribution, but effective sample sizes are modest (ESS: [112, 212]). This is expected: the trade network is sparse relative to its node count ($N=231$, density 0.087), so the likelihood surface is flat in many directions, and the nonlinear payoff-to-\bB mapping concentrates the posterior on a narrow, curved manifold that random-walk proposals traverse slowly. Importantly, low ESS affects the \textit{precision} of posterior summaries (widening credible intervals) but does not bias them. PSIS diagnostics are acceptable, with fewer than 0.6\% of observations flagging Pareto $k > 0.7$ under either model. The near-equivalence in out-of-sample predictive performance suggests that the game-theoretic parameterization adequately captures the block structure of this network without sacrificing fit relative to the unconstrained alternative, confirming our EM estimate conclusion that the trade network is well explained by the generative mechanism of the \FW. 

\input{table_results_gibbs}

\paragraph{A Sampling Signature of Hard Collapse.} We repeat the self-regularization demonstration with the Bayesian estimator, using the Zachary Karate Club network as a sample case. As before, we fit the constrained $K=3$ model with the two peripheral blocks' diagonal payoffs pinned to a common value (gauge $2/3,\,1/3,\,1/3$), making them \textit{a priori} exchangeable; if the data do not demand a distinct third type, the manifold collapses to a $K=2$ configuration, and the model should find it. Under Gibbs sampling the collapse registers differently from the EM. The split-ratio null direction (Proposition~\ref{si:manifold_collapse}) leaves the $K=3$ posterior non-identified along a flat ridge: EM escapes it by basin selection across restarts, but a random-walk sampler cannot traverse a flat ridge. The $K=3$ chains therefore exhibit catastrophic mixing failure (high $\hat R$, very low ESS) on precisely the networks where EM detects collapse. This contrasts with the $K=2$ case where chains mix cleanly ($\hat R \approx 1.0$, high ESS). The same geometric degeneracy that EM resolves cleanly is what breaks MCMC sampling.



%% file: table_results_gibbs.tex

\begin{table}[t]
\centering
\small
\begin{tabular}{%
  l                    
  l                    
  l                    
  l                    
  l                    
}
\toprule
Model &
Block Freqs &
Payoffs &
\multicolumn{1}{c}{Diagnostics} &
Verdict \\
\midrule
  Free SBM$_\text{EM}$ & X = 0.76, Y = 0.23 & & ELBO -11925.41 \\[4pt]
  
  GTBM$_\text{EM}$ & X = 0.22, Y = 0.78 & $b{=}0.73, c{=}0.22$ & ELBO -11929.33, $\Delta\Sigma_{\mathrm{BIC}}$ = -4.62 & game-theoretic (decisive) \\
  & & & $\Delta\Sigma_{\mathrm{Lap}}$ -16.94 $p_\text{Vuong}$ = 0.058 \\[4pt]
  
  GTBM$_\text{Bayesian}$ & X = 0.22 [0.22, 0.23] &
                              b = 0.78 [0.77, 0.78] & $\Delta$WAIC = +1.18  & game-theoretic (moderate) \\

                             & Y = 0.77 [0.77, 0.78] &
                              c = 0.22 [0.22, 0.23] & $\Delta$ LOOIC = -3.99 \\[4pt]

\bottomrule\\[-5pt]
\end{tabular}
    \caption{
        \textbf{Results of Bayesian analysis on the trade$_{2010}$ dataset.}
    }
    \label{tab:results_gibbs}
\end{table}

%% file: table_trade_forecasting.tex
\begin{table}[h!]
  \centering
  \newcolumntype{d}[1]{D{.}{.}{#1}}
  \begin{tabular}{
    l
    d{1.3}
    d{1.3}
    d{1.3}
    r
    c
    l
    l
    l
  }
    \toprule
    Comparison &
    \multicolumn{1}{c}{Ours MAE} &
    \multicolumn{1}{c}{Base MAE} &
    \multicolumn{1}{c}{$|\Delta M|$} &
    \multicolumn{1}{c}{\% improv.} &
    \multicolumn{1}{c}{Wins} &
    \multicolumn{1}{c}{Wilcox $p$} &
    \multicolumn{1}{c}{$t$-test $p$} \\
    \midrule
    \multicolumn{8}{l}{\itshape $t \rightarrow t+1$} \\
    Ours vs Random &
      0.0036 & 0.2831 & 0.0068 & $+98.7\%$ & 30/30 &
      $\mathbf{9.31\mathrm{e}{-10}}$ &
      $\mathbf{< 10^{-10}}$ \\
    Ours vs Inertia &
      0.0036 & 0.0038 & 0.0068 & $+4.8\%$  & 17/30 &
      $\mathbf{0.0249}$ &
      $\mathbf{0.0176}$ \\
    \midrule
    \multicolumn{8}{l}{\itshape $t \rightarrow t+2$} \\
    Ours vs Random &
      0.0051 & 0.2825 & 0.0093 & $+98.2\%$ & 29/29 &
      $\mathbf{1.86\mathrm{e}{-9}}$ &
      $\mathbf{< 10^{-10}}$ \\
    Ours vs Inertia &
      0.0051 & 0.0053 & 0.0093 & $+4.3\%$  & 19/29 &
      $\mathbf{0.0240}$ &
      $\mathbf{0.0088}$ \\
    \midrule
    \multicolumn{8}{l}{\itshape $t \rightarrow t+3$} \\
    Ours vs Random &
      0.0059 & 0.2821 & 0.0101 & $+97.9\%$ & 28/28 &
      $\mathbf{3.73\mathrm{e}{-9}}$ &
      $\mathbf{< 10^{-10}}$ \\
    Ours vs Inertia &
      0.0059 & 0.0061 & 0.0101 & $+3.2\%$  & 18/28 &
      $\mathbf{0.0314}$ &
      $0.0532$ \\
    \bottomrule\\[2pt]
    \end{tabular}
\caption{
    \textbf{Trade forecasting.} Horizon $h=1$ tests compare 30 paired $t$ vs.~$t+1$ networks, $h=2$ compare 29 paired $t$ vs.~$t+2$ networks, $h=3$ compare 28 paired $t$ vs.~$t+3$ networks.}
\label{tab:trade_forecasting}
\end{table}